\documentclass[hidelinks, 12pt]{article}
\usepackage{geometry}
\usepackage{setspace}
\usepackage[utf8]{inputenc}
\usepackage{xcolor}
\usepackage{bm}
\usepackage{amsfonts}
\usepackage{amsmath}
\usepackage{amssymb}
\usepackage{amsthm}
\usepackage{graphicx}
\usepackage{xfrac}
\usepackage{multirow}
\usepackage{pgfplots}
\usepackage{tikz}
\usepackage{dcolumn}
\usepackage{threeparttable}
\usepackage{lscape}
\usepackage{tabularx}
\usepackage{booktabs}
\usepackage{caption}
\usepackage{subcaption}
\usepackage{siunitx}
\usepackage{adjustbox}
\usepackage{accents}
\usepackage{comment}
\usepackage{epstopdf}
\usepackage{hyperref}
\usepackage[noabbrev]{cleveref}
\usepackage{tasks}
\usepackage{appendix}
\usepackage{mathtools}
\usepackage{fancyhdr}
\usepackage[round]{natbib}

\usepackage [english]{babel}
\usepackage [autostyle, english = american]{csquotes}
\MakeOuterQuote{"}
\usepackage{authblk}
\usepackage{blindtext}

\usepackage{datetime}

\newdateformat{monthyeardate}{%
  \monthname[\THEMONTH], \THEYEAR}

\usepackage{float}
\usepackage{rotating}
\usepackage{enumitem,amssymb}
\usepackage{soul}
\usepackage{bbm}

\usepackage{times}

\pgfplotsset{compat=1.15}

\usepackage{footmisc}

\newtheorem{definition}{Definition}
\newtheorem{lemma}{Lemma}
\newtheorem{proposition}{Proposition}
\newtheorem{theorem}{Theorem}

\newtheorem{assumption}{Assumption}

\newtheorem{observation}{Observation}

\theoremstyle{definition}

\theoremstyle{definition}

\theoremstyle{definition}
\newtheorem*{example*}{Example}

\theoremstyle{definition}
\newtheorem*{examples*}{Examples}

\crefname{prop}{Proposition}{Propositions}
\crefname{definition}{Definition}{Definitions}
\crefname{lemma}{Lemma}{Lemmas}
\crefname{definition}{Definition}{Definitions}
\crefname{theorem}{Theorem}{Theorems}
\crefname{proposition}{Proposition}{Propositions}
\crefname{corollary}{Corollary}{Corollaries}
\crefname{assumption}{Assumption}{Assumptions}
\crefname{assumptionp}{Assumption}{Assumptions}

\crefname{claim}{Claim}{Claims}
\crefname{section}{Section}{Sections}
\crefname{figure}{Figure}{Figures}
\crefname{exmp}{Example}{Examples}
\crefname{observation}{Observation}{Observations}

\DeclareMathOperator{\E}{\mathbb{E}}
\newcommand{\reals}{\mathbb{R}}

\DeclareMathOperator*{\argmax}{arg\,max}

\makeatletter
\DeclareRobustCommand\citepos													
  {\begingroup\def\NAT@nmfmt##1{{\NAT@up##1's}}%
   \NAT@swafalse\let\NAT@ctype\z@\NAT@partrue
   \@ifstar{\NAT@fulltrue\NAT@citetp}{\NAT@fullfalse\NAT@citetp}}

\pretocmd{\NAT@citex}{%
  \let\NAT@hyper@\NAT@hyper@citex
  \def\NAT@postnote{#2}%
  \setcounter{NAT@total@cites}{0}%
  \setcounter{NAT@count@cites}{0}%
  \forcsvlist{\stepcounter{NAT@total@cites}\@gobble}{#3}}{}{}
\newcounter{NAT@total@cites}
\newcounter{NAT@count@cites}
\def\NAT@postnote{}

\def\NAT@hyper@citex#1{
  \stepcounter{NAT@count@cites}%
  \hyper@natlinkstart{\@citeb\@extra@b@citeb}#1%
  \ifnumequal{\value{NAT@count@cites}}{\value{NAT@total@cites}}
    {\if*\NAT@postnote*\else\NAT@cmt\NAT@postnote\global\def\NAT@postnote{}\fi}{}%
  \ifNAT@swa\else\if\relax\NAT@date\relax
  \else\NAT@@close\global\let\NAT@nm\@empty\fi\fi								
  \hyper@natlinkend}
\renewcommand\hyper@natlinkbreak[2]{#1}

\patchcmd{\NAT@citex}															
  {\ifNAT@swa\else\if*#2*\else\NAT@cmt#2\fi
   \if\relax\NAT@date\relax\else\NAT@@close\fi\fi}{}{}{}
\patchcmd{\NAT@citex}
  {\if\relax\NAT@date\relax\NAT@def@citea\else\NAT@def@citea@close\fi}
  {\if\relax\NAT@date\relax\NAT@def@citea\else\NAT@def@citea@space\fi}{}{}
\patchcmd{\NAT@cite}{\if*#3*}{\if*\NAT@postnote*}{}{}
\makeatother

\title{Selection Procedures in Competitive Admission}
\author{Nathan Hancart\thanks{
Nathan Hancart: Department of Economics, University of Oslo, Eilert Sundts hus, Moltke Moes vei 31, N-0851, Oslo, Norway, nathanha@uio.no. I thank Anna Becker, Jacopo Bizzotto, Bård Harstad and Ran Spiegler for useful comments.}}
\date{\monthyeardate\today}

\begin{document}

\maketitle

\onehalfspacing
\begin{abstract}
\noindent I study how organisations choose selection procedures in a competitive environment. Two firms compete to hire candidates of unknown productivity from a common pool. Firms simultaneously post a selection procedure which consists of a test and an acceptance probability for each test outcome. After observing the firms' selection procedures, each candidate can apply to one of them. Firms can vary both the accuracy and difficulty of their test. The firms face two key considerations when choosing their selection procedure: the statistical properties of their test and the selection into the procedure by the candidates. I show that there is a unique symmetric equilibrium where the test is maximally accurate but minimally difficult. Intuitively, competition leads to maximal but misguided learning: firms end up having precise knowledge that is not payoff-relevant. In contrast, when firms face capacity constraints or have the possibility of making a wage offer, they use more difficult tests in equilibrium. I also consider asymmetric equilibria where one firm is more selective than another.
\end{abstract}

\newpage

\section{Introduction}

An organisation's admission process, whether it is a firm hiring workers or a university admitting students, usually consists of two important parts: recruitment and selection. Recruitment is the process of attracting the most suitable pool of candidates while selection aims at identifying the best candidates from that pool. 

An important factor for candidates when deciding where to apply and spend their limited resources is the probability of getting accepted. Candidates applying to college are advised to choose where to apply according to the likelihood of getting accepted.\footnote{A common advice is to balance between `safe', `match' and `reach' colleges when deciding where to apply \citep[e.g.,][a college counselling firm]{collegewise}.} Similarly, if job applications are time-consuming, job seekers must prioritise where to apply. For example, PrepLounge, an interview preparation platform for jobs in consulting and the financial sector, argues that `with your limited time, you might be better off focusing on just a handful [of firms] to maximize the quality of your applications.' \citep{preplounge}.

If candidates decide where to apply based on their chance of getting accepted, then how organisations select candidates affects the candidates' application behaviour. In other words, when designing their selection procedures, organisations need to take into account both the statistical properties of the selection procedure and its impact on the pool of candidates it attracts. This paper studies how these two elements interact in a competitive market for admission and determine the properties of the selection procedures used in equilibrium.

There are various ways in which an organisation can vary the statistical properties of its selection procedures. One way is to vary how precise its testing is. A university could include an additional interview or a firm could require an additional test to gather more information. Another possibility is to vary what type of candidate the organisation learns about. The tests can be difficult, making them effective at identifying top candidates, or easy, making them effective at identifying poor candidates.\footnote{For example, some of UCL's undergraduate programmes changed their required admission test `to further differentiate between the many highly qualified candidates who apply [...].' \citep{UCLCS}} These two dimensions of testing can be captured by the notions of accuracy \citep{lehmann1988} and difficulty \citep{hancart2024}.

This paper develops a model for understanding how firms choose selection procedures in a competitive environment. I characterise the properties of the tests used in equilibrium and show how they depend on the characteristics of the admission procedure. I vary the model along two dimensions: whether firms face capacity constraints and whether they can make wage offers. In the baseline model, admission procedures are simple: firms only need to decide whether to accept candidates and do not face capacity constraints. For example, firms could be large colleges competing for admitting students\footnote{In the US higher education context, the importance of capacity constraints varies from institution to institution. On the one hand, in most colleges, capacity has increased faster than enrolment and only 75\% of seats where used in 2019 \citep{lundy2020}. Moreover, an important concern for many colleges is the demographic decline in incoming cohorts, with fears of an `enrolment cliff' leading to intense competition \citep{bloomberg,agb}. On the other hand, elite colleges have notoriously low capacity \citep{atlantic,hoxby2009}.} or firms in a competitive labour market where demand exceeds supply. Firms want to accept any candidate with positive productivity, while all candidates want to be accepted by any firm. I show that in any symmetric equilibrium, the test used must be maximally accurate and minimally difficult. In other words, competition leads firms to use as much information as possible but learn too precisely about poor candidates compared to what would be optimal absent competition. 

I contrast this baseline model with two modifications of the admission environment. In the first one, firms face capacity constraints. In the second, firms also make a wage offer if they accept the candidate. I show that in both these variations over the baseline model, firms use more difficult tests in equilibrium. A key mechanism for these results is how the candidates' selection into the admission procedures varies across environments.

In the baseline model, two identical firms simultaneously post a selection procedure that consists of a test and an acceptance rule. A test is a Blackwell experiment that outputs a binary signal. The firms can adjust both the accuracy and the difficulty of their test. The acceptance rule is an accept/reject decision based on the signal realisation. There is a continuum of candidates that differ in their productivity. Each candidate knows his own productivity and decides where to apply after having observed the selection procedures. Applying is costly for the candidates so after having observed the selection procedures, they apply to only one of the two firms.\footnote{The constraint follows naturally when the test is a probation period or an internship, forcing the candidate to attend only one firm. Alternatively, the constraint can come from a requirement that candidates must put some effort in the selection process and can direct that effort to at most one of the two firms.} I study symmetric subgame perfect equilibria of this game. 

To characterise the equilibrium tests, I use two natural orders on experiments: accuracy \citep{lehmann1988} and difficulty \citep{hancart2024}. Accuracy is a weaker order than \citepos{blackwell1953} informativeness order for environments satisfying monotonicity assumptions. It captures a notion of precision of tests. The difficulty order was introduced in \citet{hancart2024}. This notion captures that varying the difficulty of a test changes which types are better identified: a more difficult test is informative after a high signal, as only high types are likely to produce a high signal but it is less informative after a low signal. I assume that firms can freely vary two parameters that determine the level of accuracy and difficulty: for a fixed accuracy level, increasing the difficulty level increases the difficulty of the test and vice versa. 

\cref{theo:characterisation} shows that, in the baseline model, there is a unique symmetric equilibrium where firms use a maximally accurate but minimally difficult test among the tests that provide payoff-relevant information. In other words, in equilibrium, the firms have precise information after a low signal but the high signal contains little information compared to what would be optimal in a decision problem. 

\cref{theo:characterisation} also yields comparative statics for the equilibrium tests and acceptance probabilities of candidates. When the applicant pool becomes stronger, in the likelihood ratio order sense, the equilibrium selection procedure becomes less difficult and acceptance probabilities increase for all candidates. When the testing technology improves, e.g., more accurate screening technology becomes available, candidates are accepted with higher probability on average. In the college application context, the model predicts that better students or improved testing technology lead to \textit{lower} observed selectivity. This paper therefore offers new channels to explain the decline in selectivity in college admissions in the US \citep{hoxby2009}.

In \cref{sec_extensions}, I consider two modifications to the baseline admission environment and show how they can change the predictions of the model. I first consider the case where firms face capacity constraints. Specifically, I look at the case where the mass of candidates getting a high signal is always larger than the total capacity in the market. In \cref{prop_capacity}, I show that capacity constraints lead firms to use the most difficult test in equilibrium. 

I also use the presence of capacity constraints to explore the possibility of asymmetric equilibria. In particular, I show conditions under which a \textit{two-tier structure} can emerge in equilibrium, i.e., an equilibrium where a selective firm only attracts high types and a safe firm attracts lower types. Assuming types are binary, I show that, when there are capacity constraints, it is possible to construct a two-tier structure equilibrium for some parameter values. In that equilibrium, the selective firm chooses a test that is either more accurate or more difficult than the safe firm. Therefore, ex-ante identical firms can become ex-post differentiated endogenously through their choice of selection procedure. 

In the second variation, firms can make a wage offer and do not face capacity constraints. In this model, an admission procedure consists of a test, an acceptance rule and a wage offer conditional on the realised signal. When firms can make wage offers, they compete both by using their acceptance rule and by making wage offers. I show that the test offered in any symmetric equilibrium is maximally difficult.

Wage competition and capacity constraints show how the type of admissions procedures affects the qualitative properties of tests used in equilibrium. When firms are not constrained by capacity and have fixed wages, the test used in a symmetric equilibrium is `lemon-dropping', i.e., the test is good at identifying low types. If firms face capacity constraints or compete using wages, the model predicts that firms will use `cherry-picking' tests, i.e., tests good at identifying high types.

The key mechanism driving the results is the selection of candidates into the test. I say that selection into a test is positive if whenever one candidate prefers a test over another, then all candidates with higher productivity also prefer that test. 

To prove the characterisation in the baseline model, where there is no capacity constraint nor wage offers, I show that there is positive selection into more accurate tests and into \textit{easier} tests. In the baseline model, a Bertrand-style logic drives profits to zero in equilibrium. If  profits are zero and selection into a test is positive, any deviating firm can construct a deviation that attracts only positive types. Combined with the zero profits condition, this deviation is profitable.

Intuitively, selection into a more accurate test is positive because higher types benefit more from a more precise test. To show that firms choose a minimally difficult test, I show that there is positive selection into an easier test in equilibrium. The reason is that if firms want to attract any candidate when offering a more difficult test, they must also have  a more lenient acceptance rule, otherwise no type would ever want to deviate. Low productivity candidates benefit relatively more from a more lenient acceptance rule as they are more likely to produce a low signal, and even more so in a more difficult test. Therefore, the selection into the more difficult test is negative. These two observations imply the existence and uniqueness of a symmetric equilibrium where firms use a maximally accurate but minimally difficult test.

Capacity constraints and wage offers change the mechanics through which the selection of candidates into tests operates. Consider first the case of capacity constraints. One important observation when firms are at capacity is that they cannot benefit from undercutting their competitor. Therefore, under our assumption, firms only accept after a high signal in equilibrium. If a firm deviates to an easier test, candidates face the following trade off. In the difficult test, they have a low chance of getting the high signal but a higher probability of being accepted, if they receive the high signal. In the easy test, they have a high chance of getting the high signal but a lower probability of being accepted if they receive the high signal. Relatively high types resolve this trade off in favour of the difficult test, leading to positive selection into more difficult tests. This positive selection leads firms use the most difficult feasible test in equilibrium. 

Under wage competition, I also argue that there is positive selection into more difficult tests. When firms can make wage offers, competition is fierce not on the acceptance rule but on the wages firms offer: firms do not over-accept but they overpay. Therefore, candidates only receive an offer after a high signal. As in the case of capacity constraints, higher types are willing to `take a risk' on the more difficult test in exchange of a higher wage if accepted, leading to positive selection into a more difficult test.

Finally, I note that there is an alternative interpretation of the model as a model of competing certifiers. Most in line with the payoffs of the model would be competing programmes or majors within a university. The test corresponds to exams and the acceptance decision to whether the student completes the programme. Under this interpretation, the students are motivated by having a degree and choose majors based on their likelihood of obtaining it. For example, \citet{butcher_et_al2014} and \citet{ahn_et_al2024} show the grading policy influences the student's choice of major. The majors maximise both the quality of students certified and their quantity. The model therefore provides a framework to study how university programmes adapt their testing when facing competition from other programmes.

\subsection*{Relation to the literature}

This paper develops a model of choice of selection procedures in a competitive environment. When choosing their selection procedure, the firms must consider two key channels: the statistical properties of their test and the selection into their selection procedure. I show that accuracy and difficulty are useful orders to characterise and interpret the outcome of competition under various assumptions on the admission environment. Finally, I show that selection procedures can be a channel through which ex-ante identical firms can become ex-post differently productive.

There is a small literature that studies the design of selection procedures where strategic choices from applicants play a key role. \citet{chade_et_al2014} study a competitive markets for admission procedures in the higher education context.\footnote{There are also other models of decentralised education market with a different focus, for example \citet{che_youngwoo2016} where excess enrolment is costly for colleges or \citet{che_et_al2022} where students care about prestige.} They consider a fixed testing technology and analyse a game where universities and students make their decisions simultaneously. In this paper, I endogenise the testing technology and make it an additional instrument for competition. \citet{adda_ottaviani2024} and \citet{alonso2018} are two papers that study how changing statistical properties of tests changes candidates' application behaviour. \citet{adda_ottaviani2024} study how changing the accuracy of scientific grant evaluation, in the sense of \citet{lehmann1988}, affects application behaviour.\footnote{\citet{adda_ottaviani2024} also study competition between fields as they can adjust the accuracy of grant evaluation.} \citet{alonso2018} examines the choice of selection procedure in a labour market setting with horizontally differentiated workers and wage bargaining. In his paper, workers differ in their fit for the firms and one firm's selection procedure is fixed while the other can adjust it. Both \citet{adda_ottaviani2024} and \citet{alonso2018} only consider changes to the accuracy of the test and do not consider other statistical properties like difficulty. This paper illustrates the importance of going beyond accuracy. Here, firms endogenously choose maximally accurate tests. However, if they can also adjust the difficulty of their test, they mostly learn about non-payoff relevant information.\footnote{In \citet{hancart2024}, where the difficulty notion is introduced, a privately informed agent can choose a test from a menu of tests offered by a decision-maker. The selection into test ordered by difficulty or accuracy by the agent also plays a role for determining the equilibrium actions and which equilibrium is DM-optimal.}

Conceptually, this paper relates to the literature studying competitive markets where firms compete by offering contracts or mechanisms \citep[e.g.,][]{rothschild_stiglitz1976, peters_severinov1997, peters1997, guerrieri_et_al2010, auster_gottardi2019}. This literature typically assumes that the firms can flexibly design a mechanism  subject to incentive-compatibility constraints. Instead, in this paper, the firms have a limited set of feasible tests but they do not need to satisfy any incentive-compatibility constraints to reveal information. This approach allows me to study how the statistical properties of the tests interact with the strategic choice of the agents. It is also worth noting that in the baseline model, absent any test, the firms could not elicit information in an incentive-compatible way. Therefore, the firms need hard information through tests to inform their decision. 

Finally, there is a literature on information intermediaries that study models where certifier(s) can disclose the quality of an agent, e.g., \citet{lizzeri1999, harbaugh_rasmusen2018,asseyer_weksler2024}. In a related context to this paper, there is also a literature that models education systems as intermediaries revealing information to a receiver \citep[][]{ostrovsky_schwarz2010,boleslavsky_cotton2015,bizzotto_vigier2024}. In the case where the agent is privately informed \citep[e.g.,][]{lizzeri1999, harbaugh_rasmusen2018}, selection into the certifier also plays a key role in these papers. These information design problems are however different as the intermediaries are not trying to select candidates but to collect a fee or reveal something about them. 

\section{Model}

There are two firms and a continuum of agents with mass normalised to one. Each agent has a private type which corresponds to their value to the firms $\theta\in \Theta=[\underline{\theta},\overline{\theta}]$, noting that $\underline{\theta}$ can be negative. Types are distributed according to a continuous cdf $F$ admitting a density $f$. 

The firms can decide whether to admit agents, $a\in\{0,1\}$, and an \textsl{accuracy} and \textsl{difficulty} level, $(\sigma,d)\in \Sigma\times D\subseteq\reals^2$, that determine a binary test $\tau(\sigma,d)\in \Pi:=\{\pi:\Theta\rightarrow \Delta\{l,h\}\}$. A generic test in $\Pi$ is denoted by $t$. The conditional probabilities of test $t\in \Pi$ are denoted by $\pi_t(\cdot \vert \theta)$ and I interpret signal $h$ as the high signal. To simplify notation, I denote by $\pi_t(\theta)$ the probability that type $\theta$ sends signal $h$. Denote by $\tau(\Sigma\times D)$ the set of feasible tests, i.e., the set of tests obtained by choosing some accuracy and difficulty level. The concepts of accuracy and difficulty are defined and explained in \cref{sec_orders}.

A selection procedure is a test $t\in \tau(\Sigma\times D)$, or equivalently an accuracy and difficulty level, $(\sigma,d)\in\Sigma\times D$, and a decision rule, $\alpha:\{h,l\}\rightarrow [0,1]$, a mapping from the signal to a probability of accepting: $s=(t,\alpha)$. Let $S$ be the set of selection procedures.

The firms simultaneously post a selection procedure $s\in S$. After observing the admission process, agents decide whether to apply to firm $1$ or $2$. Denote by $\phi:\Theta\times S\times S\rightarrow [0,1]$ the probability agent $\theta$ chooses firm $1$ given the selection procedures.

The agents have payoffs $u(a)=a$, i.e., they want to be accepted. Firm $1$'s payoffs are $$v(s,s',\phi)=\int_{\Theta}\phi(s,s',\theta) \theta\Big(\pi_t(\theta)\alpha(h)+(1-\pi_t(\theta))\alpha(l)\Big)dF.$$
Firm $2$'s payoffs are defined analogously. The firms care both about how many agents they attract and their quality. It also assumes there is no capacity constraint for the firm. This captures the idea that the supply of agents is smaller than the total demand and therefore the two firms must compete to attract them. I introduce capacity constraints in \cref{sec_capacity}. 

Call $T_i=\{t\in \tau(\Sigma\times D): \int_\Theta\theta\,(1-\pi_t(\theta))dF\leq 0\leq \int_\Theta\theta\,\pi_t(\theta)dF\}$ the set of \textit{minimally informative tests}. These are all the tests that generate payoff-relevant information for the firms. 

I consider subgame-perfect equilibria of this game where agents break ties uniformly and firms use pure strategies.\footnote{I consider asymmetric equilibria in \cref{sec_asy}.}

\subsection{Feasible tests}\label{sec_orders}

In this subsection, I describe in detail the set of feasible tests, $\tau(\Sigma\times D)$. 

I maintain the following assumptions throughout the paper.

\begin{assumption}\label{assum_mon} 
For each test $t\in \tau(\Sigma\times D)$, $\pi_t(h\vert \theta)$ is increasing and $\pi_t(h\vert \theta)\in (0,1)$ for any $\theta\in (\underline{\theta},\overline{\theta})$.
\end{assumption}
\cref{assum_mon} guarantees that a high signal is good news about the type no set of types with positive measure can be identified or excluded from observing any signal. 

I make use of two partial orders on tests, accuracy and difficulty, which capture vertical and horizontal properties of tests.

\begin{definition}[Accuracy and Difficulty]\label{def:acc_diff}
    A test $t$ is more accurate than a test $t'$, $t\succeq_a t'$, if for all $\theta,\theta'\in (\underline{\theta},\overline{\theta})$ with $\theta'<\theta$,\begin{align*}
        \frac{\pi_t(h\vert\theta)}{\pi_t(h\vert\theta')}\geq\frac{\pi_{t'}(x\vert\theta)}{\pi_{t'}(x\vert\theta')}\geq\frac{\pi_t(l\vert\theta)}{\pi_t(l\vert\theta')}, \;\text{for } x=l,h.
    \end{align*}

    A test $t$ is more difficult than $t'$, $t\succeq_d {t'}$, if for all $\theta,\theta'\in (\underline{\theta},\overline{\theta})$ with $\theta'<\theta$,$$
    \frac{\pi_t(x\vert\theta)}{\pi_t(x\vert\theta')}\geq \frac{\pi_{t'}(x\vert\theta)}{\pi_{t'}(x\vert\theta')}, \;\text{for } x=l,h.
    $$
\end{definition}

Examples are provided at the end of this section.

When two tests are comparable in terms of difficulty, they are not in terms of accuracy, except in knife-edge cases where one of the likelihood ratios is constant. When the state space is binary, all tests are comparable in difficulty or accuracy. In \cref{app_cost}, I also show that if we evaluate the cost of a test using continuous, posterior-separable cost function \citep{caplin_et_al2022}, for any given test, we can find another test comparable in the difficulty order with equal cost.

\citepos{lehmann1988} accuracy order captures a notion of informativeness of a test. It allows for the comparison of more experiments than \citepos{blackwell1953} order.\footnote{\citet{lehmann1988} showed that a decision maker with monotone decision preferences \citep{karlin_rubin1956} always prefers a more accurate test. \citet{quah_strulovici2009} have extended this result to preferences that form an interval dominance order family, a generalisation of both single-crossing preferences \citep{milgrom_shannon1994} and monotone decision preferences. The term accuracy comes from \citet{persico2000}.} Accuracy is a concept defined for tests satisfying the monotone likelihood ratio property with arbitrary number of signals. However, given our focus on tests with binary signals, I give a definition that is equivalent to \citepos{lehmann1988} for binary signals. The equivalence is shown in \cref{app_prelim}.\footnote{A similar notion of informativeness also appears in \citet{deb_stewart2018} without making the connection to \citet{lehmann1988}.}

The difficulty order captures the following intuitive property of difficulty: when a test is more difficult than another, a high grade shifts beliefs more towards high types in the more difficult test. This is because it is harder to receive a high grade in a difficult test. On the other hand, after a low grade, beliefs are more pessimistic in the easier test. Indeed only bad candidates are expected to fail an easy test whereas most types fail a difficult test. This intuition is formalised in the following proposition and illustrated in \cref{fig:difficulty}.

Let $\mu(\cdot\vert t,x)$ denote the posterior beliefs in test $t$ after signal $x$ and $\succeq_{FOSD}$ the first-order stochastic dominance order.

\begin{proposition}[\citet{hancart2024}]\label{prop_charac_dif}
    A test $t$ is more difficult than $t'$ if and only if $\mu(\cdot\vert t,x)\succeq_{FOSD}\mu(\cdot\vert t',x)$ for $x=h,l$ for any prior (including non full-support).
\end{proposition}

\begin{figure}
       \centering
       \begin{tikzpicture}[scale=1]

  \begin{scope}[scale = 2]

   \def\zval{3.5}

    \draw[|->, thick] (-\zval, 1) -- (\zval, 1);
    \draw[fill = black] (0, 1) circle (1pt);
    \node[below = 1ex] at (0, 1) {$\E[\theta]$};

    \draw[blue, fill=blue] (-0.75*\zval, 1) circle (1pt);
    \node[below = 1ex] at (-0.75*\zval, 1) {$\E[\theta\vert l,t']$};

\draw[blue,fill=blue] (0.25*\zval, 1) circle (1pt);
    \node[below = 1ex] at (0.25*\zval, 1) {$\E[\theta\vert h,t']$};

\draw[red,fill=red] (0.8*\zval, 1) circle (1pt);
    \node[below = 1ex] at (0.8*\zval, 1) {$\E[\theta\vert h,t]$};

\draw[red,fill=red] (-0.33*\zval, 1) circle (1pt);
    \node[below = 1ex] at (-0.33*\zval, 1) {$\E[\theta\vert l,t]$};

    \draw[->,blue,  thick] (0,1.12) arc
    [
        start angle=20,
        end angle=160,
        x radius=\zval*0.395,
        y radius =1.1
    ] ;
 
 \draw[->,blue,  thick] (0,1.12) arc
    [
        start angle=160,
        end angle=20,
        x radius=\zval*0.125,
        y radius =0.4
    ] ;

    \draw[->,red,  thick] (0,1.12) arc
    [
        start angle=20,
        end angle=160,
        x radius=\zval*0.17,
        y radius =0.4
    ] ;

     \draw[->,red,  thick] (0,1.12) arc
    [
        start angle=160,
        end angle=20,
        x radius=\zval*0.42,
        y radius =1.1
    ] ;

  \end{scope}

\end{tikzpicture}%
\caption{Illustration of posterior means for two tests, $t\succeq_d t'$. The good news signal $h$ shifts the posterior towards a higher posterior mean in the more difficult test. The bad news signal $l$ shifts the posterior towards lower posterior mean in the easier test.}
       \label{fig:difficulty}
   \end{figure}
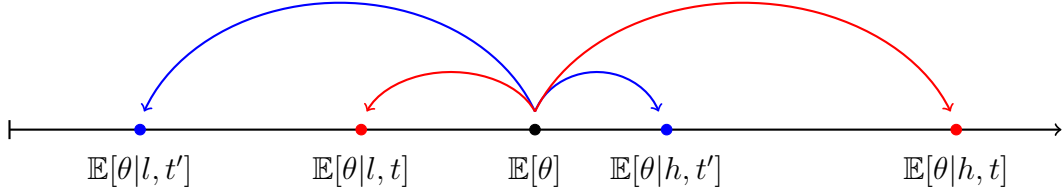 

We can now formalise the notion of accuracy and difficulty level as follows: \begin{align*}
        &\text{for all }\sigma\in \Sigma,\; \tau(\sigma,d)\succeq_d \tau(\sigma,d')\Leftrightarrow d\geq d',\\
        &\text{for all }d\in D,\; \tau(\sigma,d)\succeq_a \tau(\sigma',d)\Leftrightarrow \sigma\geq \sigma'.
    \end{align*} 

Intuitively, starting from a test $\tau(\sigma,d)$, increasing $\sigma$ increases the accuracy and increasing $d$ increases the difficulty. Note that generically, if $d\neq d'$, then $\tau(\sigma,d)$ and $\tau(\sigma,d')$ are not comparable in accuracy. Similarly if $\sigma\neq \sigma'$, then $\tau(\sigma,d)$ and $\tau(\sigma',d)$ are not necessarily comparable in difficulty.

Abusing notation, I write $\pi_{\sigma,d}(\theta)$ for the probability of sending a high signal in test $\tau(\sigma,d)$. I provide examples of tests comparable in terms of accuracy and difficulty below. 

\begin{examples*}
    For simplicity, I normalise $\Theta$ to $\Theta=[0,1]$. In each example, accuracy is increasing in $\sigma$ and difficulty in $d$.\begin{enumerate}
     \item Let $\pi_{\sigma,d}(\theta)=\sigma\pi(\theta)+(1-\sigma)(1-d)$ where $\pi:\Theta\rightarrow (0,1)$ is increasing and $\sigma,d\in[0,1]$. 
     \item Let $\pi_{\sigma,d}(\theta)=\big(\frac{1}{2}+\sigma(\theta-\frac{1}{2})\big)^d$ where $\sigma,d\in[0,1]$. 
     \item Let $\pi_{\sigma,d}(\theta)=\text{Pr}[y\geq d\vert\theta]$ where  $y=\theta+ \frac{\epsilon}{\sigma}$ and $\epsilon\sim U(-1/2,1/2)$ with $1-\frac{1}{2\sigma}\leq d\leq \frac{1}{2\sigma}$.
\end{enumerate}
\end{examples*}

The first example can be interpreted as follows. With probability $\sigma$, the test is informative and the probability of a high signal depends on $\theta$. With the complement probability, all types receive a high signal with probability $1-d$. In the second example, the test is again more sensitive to the type and difficulty reduces the probability of receiving a high signal. In the last example, there is an underlying continuous signals $y$. The test is passed when the underlying signal is above the threshold $d$. The higher the threshold, the more difficult the test. The restrictions guarantee that $\text{Pr}[y\geq d\vert\theta]$ gives an interior probability for all types.\footnote{In general, with a test $\pi_{\sigma,d}(\theta)=\text{Pr}[y\geq d\vert \theta]$ with $y=\theta+\frac{\epsilon}{\sigma}$ and $\epsilon\sim G$, the difficulty increases in $d$ when $G$ and $1-G$ are log-concave. These conditions are satisfied by many common distributions like the normal, logistic or Laplacian distributions. The accuracy does not necessarily increase with $\sigma$. If the signal $y$ becomes too precise, the signal $\pi_{\sigma,d}$ just indicates whether the type is above or below $d$. This would be a useless signal if, e.g., the prior is concentrated above $d$. In that case, a DM might prefer a test with a lower $\sigma$.}

Intuitively, accuracy increases when the probability of receiving a high signal is more sensitive to the type. Difficulty increases when obtaining the higher signal is harder for all types, but more so for lower types.

I also record the following observation.

\begin{lemma}\label{lemma_diff}
    Suppose test $t$ is more difficult than $t'$. Then for all $\theta\in (\underline{\theta},\overline{\theta})$, $\pi_t(\theta)\leq \pi_{t'}(\theta)$. 
\end{lemma}

All proofs are in \cref{app_proof}.

\cref{lemma_diff} shows that the difficulty order leads to the natural property that high signals are less likely in a more difficult test.

Finally, I maintain the following assumptions throughout. 

\begin{assumption}\label{assump:technical}
    (1) The sets $\Sigma$ and $D$ are closed intervals. (2) The probability $\pi_{\sigma,d}(\theta)$ is continuous in $(\sigma,d)$ for all $\theta\in \Theta$. (3) If two tests $t,t'\in \tau(\Sigma\times D)$ are comparable in difficulty, they are not comparable in accuracy.
\end{assumption}

The first two assumptions are technical assumptions to guarantee equilibrium existence. The last one simplifies the exposition.

\subsection{Discussion}\label{sec:discussion}

\paragraph{Cost of applying} In this model, the constraint on the application strategy of the agents is that they can apply to only one firm. This constraint can be interpreted in different ways and is natural in a number of markets. The constraint can come from the nature of the test. If the test is a probation period or an internship, the test can be performed at at most one firm. In the interpretation of the model as competing majors or programmes within a university, students can participate to only one programme. Alternatively, if applying requires effort to tailor the application package or prepare for firm-specific tests, candidates need to prioritise effort towards a subset of firms. Finally, there can be institutional constraints. For example, in the UK, applicants for undergraduate programmes can apply to up to five different programmes \citep{UCAS}. 

The results would be qualitatively the same if there was a fixed cost of applying to each firm and that cost is high enough so that no agent applies to two firms. The key assumption is that agents cannot apply to all firms in the market. A key challenge when candidates can potentially apply to all firms is that application behaviour can be non-monotonic: the candidates applying to both firms may or may not be the most productive candidates. 

\paragraph{Restrictions on tests} I make two substantive assumptions on the class of feasible tests. First, all tests have binary signals. Second, from any given test, one can find a test that is more or less accurate and more or less difficult. Binary signals would be without loss of generality if firms could choose any tests in $\Pi=\{\pi:\Theta\rightarrow\Delta X\}$ for some signal space $X$ with $\vert X\vert\geq 2$. Then, standard arguments would show that $X=\{h,l\}$ is without loss of generality and that each signal is an action recommendation, i.e., $\alpha(h)=1$ and $\alpha(l)=0$.

If the design of test is fully flexible, then there is an equilibrium where both firms offer a test with $\pi_t(h\vert \theta)=1$ when $\theta\geq 0$ and $\pi_t(h\vert\theta)=0$ when $\theta<0$, i.e., the test perfectly reveals whether the type is below or above zero. This selection procedure is exactly the same one as the one a monopolist would choose. This extreme result motivates putting restrictions on the tests available to the firms.  

The restriction on the feasible tests relies on the statistical properties of the tests directly. The goal of these restrictions is to offer an interpretable framework that allows for characterisations of the equilibria and comparative statics, while at the same time not putting direct parametric assumptions on the feasible tests. When state is binary, any test is comparable in either difficulty or accuracy. In this case, the characterisation of \cref{theo:characterisation} applies to arbitrary set of feasible tests with binary signals (up to some technical assumptions for existence).  

To make progress beyond binary signals, one would need to develop a theory of difficult tests beyond binary signals. For accuracy, the analysis extends to arbitrarily many signals. Note also that in the case of wage offers, binary signals would not be without loss of generality even in the flexible test design setup.

A strength of binary signals is that all signals have an unambiguous interpretation as either a high or a low signal. This would not be the case with more signals. This is especially important when comparing selection procedures where the tests differ in their difficulty. As we will see in the analysis, the selection into a more difficult tests varies depending on whether the low signal is rewarded or not, i.e., $\alpha(l)>0$ or $\alpha(l)=0$. The clean distinction between high and low signals allows for a transparent interpretation of the selection effects into more difficult test.

\section{Analysis}

I first show that competition leads firms to `over-accept' candidates in the sense that they reward the low signal even though it has negative posterior expected productivity. In any symmetric equilibrium, the payoffs of the firms are no different from their payoffs if they could not observe any signals.

\begin{lemma}\label{prop_bertr}
    In any symmetric equilibrium $s=(t,\alpha)$, firms' payoffs are $\max\{0,\frac{1}{2}\E[\theta]\}$.
    \begin{itemize}
        \item If $\E[\theta]\geq 0$, $\alpha(h)=\alpha(l)=1$.
        \item If $\E[\theta]<0$ and $\int_\Theta \theta\pi_t(\theta)dF>0$, $\alpha(h)=1$, $\alpha(l)> 0$.
    \end{itemize}
\end{lemma}

The intuition for \cref{prop_bertr} is the familiar Bertrand undercutting logic. If firms make positive profits, they can relax their acceptance rule and attract all candidates. If the increased probability of acceptance is small enough, the firm's payoffs are larger than when sharing the market with the other firm. Therefore in equilibrium, firms accept candidates until they make zero profits or they accept all candidates applying. This means that the firms make exactly the same profits as if they could not collect any information about candidates.

If $\E[\theta]\geq 0$, both firms accept all types and make weakly positive profits. Therefore, any test with $\alpha(h)=\alpha(l)=1$ is an equilibrium. In the rest of this section, I therefore focus on the case with $\E[\theta]<0$.

I now turn to the characterisation of the tests used in equilibrium. 

\begin{theorem}\label{theo:characterisation}
    Suppose $\E[\theta]<0$ and there is $t\in \tau(\Sigma\times D)$ such that $\int_\Theta \theta\pi_t(\theta)dF>0$. Then there is a unique symmetric equilibrium $s=\left(\tau(\sigma^*,d^*),\alpha\right)$. The accuracy and difficulty levels satisfy\begin{align*}
        &\sigma^*=\max \{\sigma\in\Sigma\},\\
        &d^*=\min\{d\in D: \int_{\Theta}\theta\pi_{\sigma^*,d}(\theta)dF\geq 0\}.
    \end{align*}
    The acceptance rule satisfies $\alpha(h)=1$ and $\alpha(l)\geq 0$.
\end{theorem}

\cref{theo:characterisation} shows that in the unique symmetric equilibrium, the test used is maximally accurate but minimally difficult, as long as it reveals some payoff-relevant information. It is illustrated in \cref{fig:two}. The selection procedures used in equilibrium are therefore suboptimal for the firms. in the sense that a monopolist acting on behalf of the two firms would choose a different strategy. In particular, it would choose a higher difficulty level and would not accept after a low signal.\footnote{Note that the monopolist would not necessarily maximise difficulty either as this could entail too few accepted candidates. For a monopolist, the choice of difficulty level trades off the expected productivity conditional on the high signal with the unconditional probability of accepting candidates.}

\begin{figure}[h!]
   
    \centering
   
    \begin{subfigure}[t]{0.47\textwidth}
    \centering
   \includegraphics[width=\textwidth]{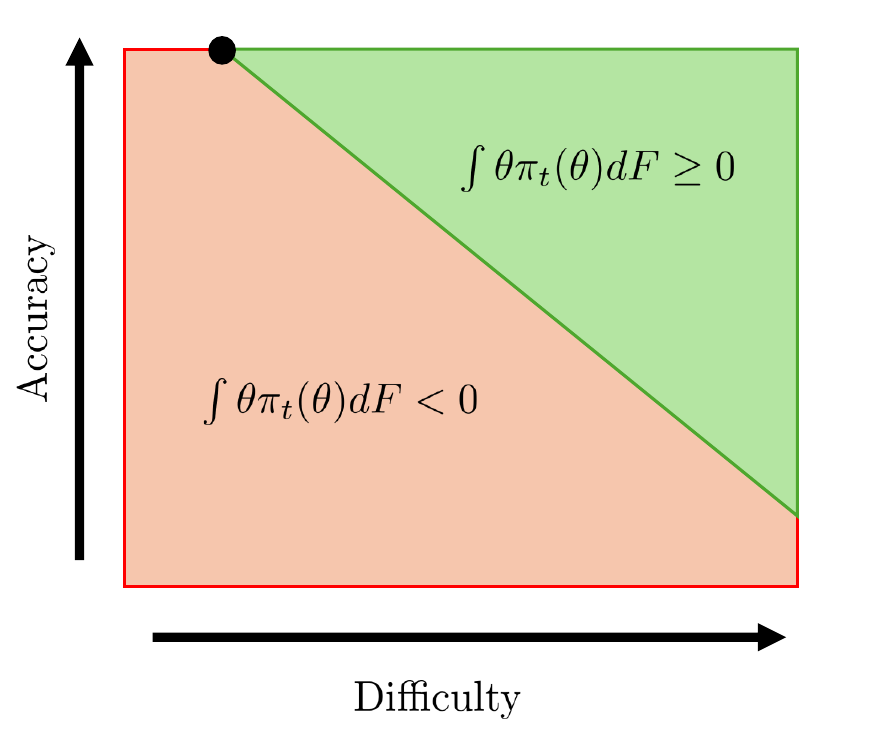}
        \caption{$\E[\theta\vert\tau(\sigma^*,d^*),h]=0$.}
    \end{subfigure}
    \hfill
    \begin{subfigure}[t]{0.47\textwidth}
        \centering
        \includegraphics[width=\textwidth]{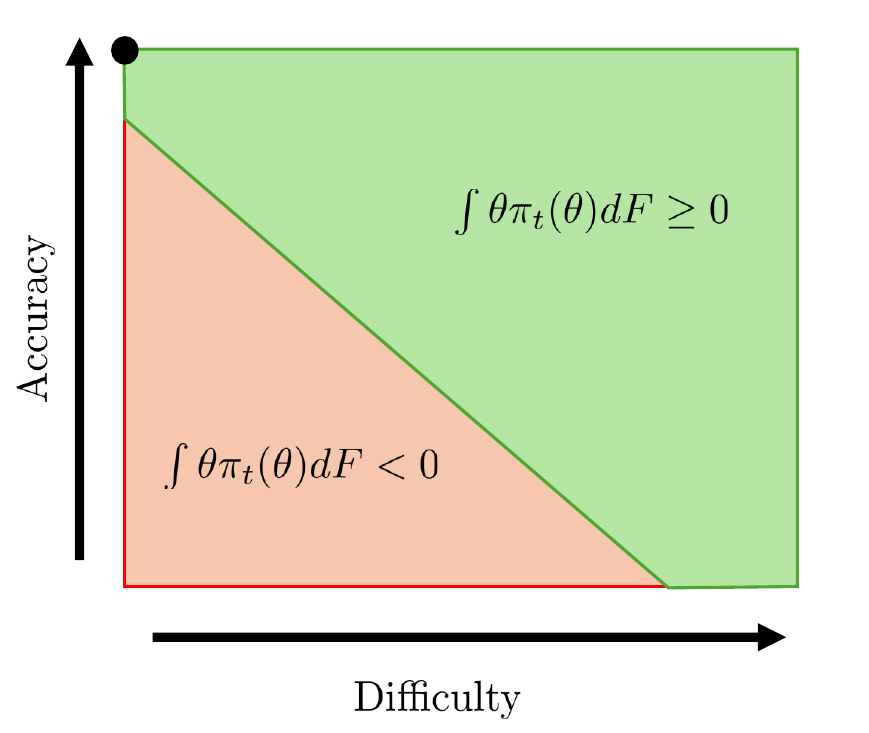}
        \caption{$\E[\theta\vert\tau(\sigma^*,d^*),h]>0$.}
    \end{subfigure}
\caption{Each point in the rectangle represents a test. A test is minimally informative if it is accurate or difficult enough. The black dot indicates the test used in symmetric equilibrium.}
\label{fig:two}
\end{figure}

If the set of feasible tests is sufficiently rich, at the equilibrium test, $\int_{\Theta}\theta\pi_{\sigma^*,d^*}(\theta)dF=0$. Therefore, in equilibrium, we have $\alpha(h)=1$ and $\alpha(l)=0$. Under that strategy, the firms best reply to their signals in equilibrium. In that case, the equilibrium we have found is also an equilibrium of a game with an alternative timing where firms do not commit to an acceptance rule but decide whether to accept only after having seen the test result. 

While accuracy and difficulty are partial orders for an arbitrary $\Theta$, they are not when the state is binary, $\vert\Theta\vert=2$. In this case, any two binary tests are comparable in difficulty or in accuracy (or both). \cref{theo:characterisation} therefore also gives a characterisation of the equilibrium tests for arbitrary restrictions on the feasible tests when the state is binary. The equilibrium test is the least difficult test among the tests that are undominated in accuracy and provide payoff-relevant information.\footnote{Binary types can be integrated directly in the formalism of this paper by having a $\pi_t(\theta)$ being constant above and below some threshold type $\theta^*$ for all tests $t$. In this case, the types below $\theta^*$ are indistinguishable and similarly for the types above $\theta^*$. Therefore, the types are effectively binary.}

\cref{theo:characterisation} follows from the selection into tests. I first show that there is positive selection into a more accurate test. Intuitively, higher productivity agents benefit relatively more from a more accurate test regardless of the strategy employed. Because any candidate equilibrium has zero profits, any deviation that attracts only positive types is profitable.

I also show that there is positive selection into an \textit{easier} test. Consider the case where both firms use the easiest minimally informative test in equilibrium and consider a deviation to a harder test. To make this deviation successful, firms must make their acceptance rule more lenient as no candidate would want to choose a more difficult with harsher acceptance rule. But when raising the acceptance probability at the low signal in the more difficult test, the agents most likely to benefit from that selection procedure are low productivity agents. Indeed, they are the agents most likely to generate low signals. Therefore, the selection into the more difficult \textit{and} more lenient acceptance rule is negative. 

When the set of feasible test is rich enough, the equilibrium test has $\E[\theta\vert \tau(\sigma^*,d^*),h]=0$, i.e., the posterior expectation at the high signal is zero. In equilibrium, the firms set $\alpha(l)=0$, i.e., a low signal is no longer rewarded. If a firm deviates to an easier test, it would need to use a harsher decision rule: if not, all types would choose the deviating firm but in the easier test, the expected payoffs are negative. Therefore, the acceptance rule from a profitable deviation must only reward the high signal. But in this case, the selection into an easier test when only the high signal is rewarded is negative as high productivity candidate are more likely to generate a high signal and are therefore more inclined to choose the harder test. 

\paragraph{Comparative statics}

The characterisation of the equilibrium allows for some comparative statics on the payoffs of the players and properties of equilibrium tests.

Let $\alpha^*(F,\Sigma\times D)$, $\sigma^*(F,\Sigma\times D)$ and $d^*(F,\Sigma\times D)$ be the equilibrium acceptance rule and accuracy and difficulty level when the primitives are $(F,\Sigma\times D)$. Similarly, let $p^*(\theta;F,\Sigma\times D)$ be the acceptance probability for type $\theta$ induced by $\alpha^*,\sigma^*,d^*$.

The following results do two comparative statics exercises. In the first one, I compare the acceptance probabilities and difficulty level when the distribution of types experiences a shift in the likelihood ratio order. In the second, I compare the acceptance probabilities when the testing technology expands. 

We say that $F'$ is higher in the likelihood ratio order than $F$ if for all $\theta'>\theta$,$$
f'(\theta')f(\theta)\geq f'(\theta)f(\theta').
$$

\begin{proposition}\label{prop:comparative_FOSD}
    Suppose the conditions of \cref{theo:characterisation} are satisfied.

    If $F'$ is higher in the likelihood ratio order than $F$, then, \begin{align*}
        &d^*(F',\Sigma\times D)\leq d^*(F,\Sigma\times D),\\
        \text{for all $\theta$, }\,&p^*(\theta;F',\Sigma\times D)\geq p^*(\theta;F,\Sigma\times D).
    \end{align*}
\end{proposition}

The first result shows that when the candidates become better, in the likelihood ratio order sense, the equilibrium tests get easier and in turn candidates are accepted with higher probability. This result follows from the fact that in equilibrium, the firms decrease the difficulty level until the expected productivity at the high signal is zero:$$
\int \theta\pi_{\sigma^*,d^*}(\theta)dF=0.
$$
As the distribution of productivity experiences a likelihood ratio order shift, the expected productivity at the high signal becomes strictly positive. To re-establish the equilibrium conditions, the difficulty level must decrease, and therefore the acceptance probabilities increase. 

\begin{proposition}\label{prop:comparative_tech}
     Suppose the conditions of \cref{theo:characterisation} are satisfied.
     
     If $\Sigma\times D\subseteq \Sigma'\times D'$, then,$$
     \E[p^*(\theta;F,\Sigma'\times D')]\geq \E[p^*(\theta;F,\Sigma\times D)].
     $$
    Furthermore, if $p^*(\theta;F,\Sigma'\times D')\geq p^*(\theta;F,\Sigma\times D)$ then $p^*(\theta';F,\Sigma'\times D')\geq p^*(\theta';F,\Sigma\times D)$ for all $\theta'>\theta$.
\end{proposition}

When the firms' set of feasible tests expands, the average acceptance probability of the candidate increases. Unlike in \cref{prop:comparative_FOSD}, the comparative statics result only holds in expectation. Therefore, it is possible that some types suffer from a change in the testing technology. However, if some types suffer from a change in technology, they must be relatively low in the sense that all types above some cutoff type benefit from a change in the technology. 

To understand this result, observe that, in equilibrium, the firms use a smaller difficulty level and higher accuracy level under $\Sigma'\times D'$ than under $\Sigma\times D$. For simplicity, assume that in both equilibria, $\alpha(l)=0$ and $\alpha(h)=1$ and denote by $(\sigma,d)$ and $(\sigma',d')$ the accuracy and difficulty levels under $\Sigma\times D$ and $\Sigma'\times D'$. The ordering of the accuracy and difficulty levels implies that \begin{equation}\label{eq:comp_exp}
    \E[\theta\vert l,\tau(\sigma,d)]\geq \E[\theta\vert l,\tau(\sigma',d')].
\end{equation} 

At the same time, in both cases, the zero profits condition must hold, $\int\theta\pi_{\sigma,d}(\theta)dF=\int\theta\pi_{\sigma',d'}(\theta)dF$, which implies\begin{equation}\label{eq:comp_zero}
     \int\theta(1-\pi_{\sigma,d}(\theta))dF=\int\theta(1-\pi_{\sigma',d'}(\theta))dF.
\end{equation}
The only way both (\ref{eq:comp_exp}) and (\ref{eq:comp_zero}) hold together is if the probability of getting a low signal is higher under $\tau(\sigma,d)$ than under $\tau(\sigma',d')$. This in turn implies that the probability of getting the high signal is \textit{lower} under $\tau(\sigma,d)$ than under $\tau(\sigma',d')$.

The single-crossing condition on the acceptance probability follows from the same argument that establishes a positive selection into easier or more accurate tests. 

\section{Capacity constraints and Wage competition}\label{sec_extensions}

\subsection{Capacity constraint}\label{sec_capacity}

As in the analysis of markets with Bertrand competition, capacity constraints can radically change the predictions of our model. The reason is that capacity constraints shift the focus of firms from both the quality and the quantity of applicants to quality only. This gives incentives to use more difficult tests in equilibrium. To focus the analysis, I only consider the case where the firms can adjust the difficulty of their tests but not the accuracy, i.e., $\Sigma=\{\sigma\}$.

I model capacity constraints as follows. Each firm has a capacity $k>0$, i.e., in equilibrium we must have $\int_{\Theta}\phi(s,s',\theta)\left(\alpha(l)(1-\pi_t(\theta))+\alpha(h)\pi_t(\theta)\right)dF\leq k$. 

The timing of the game is as follows: \begin{enumerate}
    \item Firms post selection procedures simultaneously
    \item Candidates choose where to apply
    \item Candidates applying to firm $i$ form a queue whose order is random. Firm $i$ treats applications sequentially until it exhausts the pool of applicants or hits its capacity constraint. 
\end{enumerate}

The solution concept is still symmetric subgame-perfect equilibrium where firms use pure strategies and agents break ties symmetrically.

Given the strategy of the agents $\phi$, let $p_1=\min\{1,\frac{k}{\int_{\Theta}\phi(s,s',\theta)\left(\alpha(l)(1-\pi_t(\theta))+\alpha(h)\pi_t(\theta)\right)dF}\}$ be the probability of having a candidate's application considered in firm $1$. The payoffs of firm 1 in equilibrium are$$
p_1\cdot \int_{\Theta}\phi(s,s',\theta)\theta\,\left(\alpha(l)(1-\pi_t(\theta))+\alpha(h)\pi_t(\theta)\right)dF.
$$

The payoffs of firm $2$ are defined analogously. 

I illustrate the effect of capacity constraints in the case where the mass of candidates getting a high signal is always higher than the total capacity in the market, i.e., for all $d\in D$, $\int_\Theta \pi_d(\theta)dF\geq 2k$. (Recall that in this section the firms does not adjust the accuracy level and therefore I identify a test with its difficulty level.)

\begin{proposition}\label{prop_capacity}
    Let $\overline{d}=\max D$ and assume that $2k\leq \int_\Theta\pi_{d}(\theta)dF$ for all $d\in D$.

    There is a symmetric equilibrium $s=(\tau(\overline{d}),\alpha)$, i.e., both firms offer the most difficult test. 
\end{proposition}

We get the opposite prediction to the case with no capacity constraints. The reason is that capacity constraints limit the scope for undercutting the competitor: once the firms are at capacity, there is no benefit to increasing the acceptance probability to attract all the candidates. Given the assumption on the capacity constraint, there is no incentive to accept after a low signal. When the low signal is not rewarded, the selection into the easier test is negative. Moreover, only the expected productivity of the accepted candidates matters for the payoffs as the firms are at capacity. From \cref{prop_charac_dif}, a more difficult test always has a higher posterior expected value. Therefore a deviation to an easier test results both in negative selection and lower expected value, even absent any selection effect. 

\cref{prop_capacity} gives us the following (informal) comparative statics result:

\begin{observation}
    Firms facing capacity constraints use more difficult tests in equilibrium than firms not facing capacity constraints. 
\end{observation}

This result predicts that whether the firms face capacity constraints or not will affect the qualitative nature of the tests they use. Loosely speaking, a firm facing a tight capacity constraint will use a difficult selection procedure, a `cherry-picking' type of test whereas a firm without capacity constraint will use easier tests, a `lemon-dropping' type of test.  

\subsection{Asymmetric equilibria}\label{sec_asy}

A natural question in this context is whether different selection procedures can generate differentiation among firms despite being ex-ante identical. In particular, I will look at whether a two-tier structure can emerge in equilibrium where one firm attracts all candidates above a threshold and the other firm attracts candidates below the threshold. To answer this question, I specialise the setting to a binary type model. I show that a two-tier structure never emerges when there are no capacity constraints but that it can emerge when there are. 

The set $\tau(\Sigma\times D)$ \textit{induces binary types} if there is $\theta^*$ such that for all $t\in \tau(\Sigma\times D)$,  $\pi_t(\theta)$ is constant above and below $\theta^*$. Whenever $\tau(\Sigma\times D)$ induces binary types, then, abusing notation, I write $\underline{\theta}=\E[\theta\vert\theta\leq \theta^*]$ and $\overline{\theta}=\E[\theta\vert\theta>\theta^*]$ and $\pi_t(\overline{\theta})$ and $\pi_t(\underline{\theta})$ for the probability of generating the signal $h$ for types above and below $\theta^*$ respectively. I denote by $\mu$ the mass of $\theta>\theta^*$.

A \textit{two-tier equilibrium} is an equilibrium where all types above a threshold apply to the same firm and only these types do. I call the firm where only types above the threshold apply the \textit{selective firm} and the firm where types below the threshold apply the \textit{safe firm}.\footnote{With two types, all equilibria are payoff equivalent to either a two-tier equilibrium or an equilibrium where firms use the same admission procedure.}

As the equilibrium construction in \cref{prop_two} below will use mixed strategies, I no longer require that agents break ties uniformly. I do require, however, that any types generating the same distribution over signals use the same strategy. If there are capacity constraints, the timing and rationing is as defined in \cref{sec_capacity}.

\begin{proposition}\label{prop_two}
    Suppose $\tau(\Sigma\times D)$ induces binary types. Let $\overline{t}\in \argmax_{t\in \tau(\Sigma\times D)}\frac{\pi_t(\overline{\theta})}{\pi_t(\underline{\theta})}$.

    If there are no capacity constraints, then there are no two-tier equilibria.

    If there are capacity constraints, there is a two-tier equilibrium where the selective firm chooses $\overline{t}$ and the safe firm chooses some other test $t\in \tau(\Sigma\times D)$ if\begin{align*}
        &\underline{\theta}\geq 0,\\
        &(1-\mu)\pi_{\overline{t}}(\underline{\theta})\geq\mu\pi_{\overline{t}}(\overline{\theta}),\\
        &\frac{\mu\pi_{\overline{t}}(\overline{\theta})+(1-\mu)\pi_{\overline{t}}(\underline{\theta})}{2}\geq k,\\
        \text{and }\;&\frac{\pi_t(\underline{\theta})}{2}\geq k.
    \end{align*} 
    In equilibrium, $\overline{\theta}$ chooses the selective firm and $\underline{\theta}$ chooses the safe firm with probability $\frac{\mu\pi_{\overline{t}}(\overline{\theta})+(1-\mu)\pi_{\overline{t}}(\underline{\theta})}{2(1-\mu)\pi_{\overline{t}}(\underline{\theta})}$.
\end{proposition}

The reason two-tier equilibria cannot exist without capacity constraints is that as long as the safe firm makes positive profits, the selective firm has an incentive to lower its standards to attract more candidates. Here decreasing standards corresponds to offering a selection procedure with a lower ratio of likelihood of acceptance between high and low types. On the other hand, if the selective firm is at capacity, the benefits from decreasing standards could be limited. 

With capacity constraints, the selective firm uses the test with the highest likelihood ratio at the high signal. That means that the test is either more accurate or more difficult than the other feasible tests. In equilibrium, the selective firm only accepts after a high signal. Under this strategy, the selection into the most difficult or accurate test is positive. Moreover, because in equilibrium firms are at capacity, they cannot improve payoffs by simply lowering their standards and attracting more types. For example, when the selective firm decreases its standards, it increases the share of lower quality students applying to it, thereby decreasing its payoffs.

The equilibrium I construct has high types choosing the selective firm and low types mixing between the selective and the safe firm. The sufficient conditions in \cref{prop_two} reflect the equilibrium conditions to maintain that equilibrium. The first one, $\underline{\theta}\geq 0$ is necessary to make sure that the safe firm makes profits in equilibrium. The second condition ensures that the mixed strategy is feasible. The last two conditions guarantee that the capacity constraints of both firms are binding.

\subsection{Wage competition}\label{sec_wage}

In this subsection, I consider the consequences of wage setting for the choice of equilibrium tests. Formally, a firm can offer a positive transfer to the agent based on the signal it received: $m:\{h,l\}\rightarrow \reals_+$. An admission procedure is a test, a decision rule and a transfer rule, $s=(t,\alpha,m)$. An agent's payoff is the transfer $a\cdot m$. Firm $1$'s payoffs are $$
v(s,s',\phi)=\int_{\Theta}\phi(s,s',\theta)\Big(\pi_t(\theta)\alpha(h)(\theta-m(h))+(1-\pi_t(\theta))\alpha(l)(\theta-m(l))\Big)dF.
$$

As in \cref{sec_capacity}, I focus the analysis on the case where the firms can only adjust their difficulty level, i.e., $\Sigma=\{\sigma\}$. Therefore, I identify a test with its difficulty level.

\begin{proposition}\label{prop_money}
    Let $\overline{d}=\max D$ and suppose there is $d\in D$ such that $\int_\Theta\theta\pi_d(\theta)dF>0$. In any symmetric equilibrium $s=(\tau(\overline{d}),\alpha)$ and $\alpha(l)=0$.
\end{proposition}

The main consequence of wage setting is that competition moves from the acceptance rule to the wage offered. Instead of over-accepting, i.e., setting $\alpha(l)>0$, firms compete on the wage offered, conditional on receiving a high signal. This changes the selection effect of offering a more difficult test. Now, high productivity agents benefit relatively more from a more difficult test as they are relatively more likely to get a high signal. Because of this positive selection into a more difficult test, firms can always deviate if there is a more difficult test available. 

In this section we face a similar challenge as other models of competition with asymmetric information \citep[as in e.g.,][]{rothschild_stiglitz1976} and we cannot generally establish that there is an equilibrium where the most difficult test is chosen. A deviating firm to an easier test faces negative selection but can compensate by offering lower wages. Therefore whether a deviation is profitable depends on the specification of the feasible tests and the prior. For example, I show that if $\theta\sim U[0,1]$ and $\pi_d(\theta)=\theta^d$ with $D=[0,1]$, then an equilibrium exists with $d=1$ (see \cref{app:example}).

This second variation gives us a second informal comparative statics result:\begin{observation}
    If firms compete using wages, they use more difficult tests in equilibrium than when they can only compete using admission probability. Moreover, the hiring probability is lower for almost all types.
\end{observation}
As in the case of capacity constraints, firms offer a more difficult test in equilibrium when they can compete using wage offers. Because in equilibrium they only accept after a high signal, it also implies that the probability of accepting any given type decreases.\footnote{Under capacity constraints, the probability of being accepted decreases as well but that was exogenously imposed by the capacity constraint.}

\section{Conclusion}

I have introduced a new model of competition where firms compete by posting selection procedures. The key channel I explored is how statistical properties of the tests imply different strategic choices from the tested agents. In particular, I showed that two natural orders on tests, accuracy and difficulty, create single-crossing utility differences for the agent. This led to positive or negative selection into a test that in turn determined the equilibrium.

The model makes some predictions about the qualitative nature of the tests used in equilibrium depending on the primitives of the game. In the absence of capacity constraints, the firms use the maximally accurate but easiest test that is minimally informative. We can interpret this as maximal but misguided learning. In equilibrium, firms are very confident that the candidate is of low quality after a low signal but their posterior expectation is barely high enough to make them accept the agent after a high signal. On the other hand, when firms face capacity constraints or can compete using wages, they use more difficult tests in equilibrium.

I see this model as a first step towards studying the effect of competition on the choice of tests. There are many natural extensions one would want to consider such as differentiated firms, both horizontally and vertically. Another interesting extension would be introducing peer effects which is particularly relevant in a university admission context.

\newpage

\bibliographystyle{agsm}
\bibliography{bib.bib} 

\newpage

\appendix

\section{Proofs}\label{app_proof}

\subsection{Preliminary results}\label{app_prelim}

\citet{lehmann1988} defined his notion of accuracy as follows. Take a compact signal space $\Tilde{X}\subset\reals$ and let $F_t(\cdot\vert\theta)$ be the conditional cdf of test $t$. \citet{lehmann1988} shows how information structures with discrete signal spaces can be rewritten as information structures with continuous signal spaces. 

\begin{definition}[\citet{lehmann1988}]
    A test $t$ is more accurate than a test $t'$ if $$
    x^*(\theta,x)\text{, the solution to } F_t(x^*\vert\theta)=F_{t'}(x\vert\theta),
    $$is weakly increasing in $\theta$ for each $x\in \Tilde{X}$.
\end{definition}

The following result shows the equivalence between \citepos{lehmann1988} definition and the one in \cref{sec_orders}.

\begin{proposition}
    Suppose the signal space is binary. A test $t$ is more accurate than $t'$ if and only if for all $\theta>\theta'$,$$
    \frac{\pi_t(h\vert\theta)}{\pi_t(h\vert\theta')}\geq\frac{\pi_{t'}(h\vert\theta)}{\pi_{t'}(h\vert\theta')} \quad \text{ and }\quad \frac{\pi_{t'}(l\vert\theta)}{\pi_{t'}(l\vert\theta')}\geq\frac{\pi_t(l\vert\theta)}{\pi_t(l\vert\theta')}.
    $$ 
\end{proposition}
   
\begin{proof}
   \citet{adda_ottaviani2024} show that $t$ more accurate than $t'$ is equivalent to having for all $\theta>\theta'$,$$
   F_t(F_t^{-1}(q\vert\theta')\vert\theta)\leq F_{t'}(F_{t'}^{-1}(q\vert\theta')\vert\theta),
   $$for all $q\in[0,1]$.

   Let $\Tilde{X}=[0,1]$. We can rewrite an information structure with binary signals where the probability of a high signal is $\pi_t(\theta)$ as$$
F_t(x\vert\theta)=\begin{cases}
    2(1-\pi_t(\theta))x&\text{if }x<1/2,\\
    1+2\pi_t(\theta)(x-1)&\text{if }x\geq1/2.
\end{cases}
$$

The inverse is$$
F_t^{-1}(q\vert\theta)=\begin{cases}
    \frac{q}{2(1-\pi_t(\theta))}&\text{if }q<1-\pi_t(\theta),\\
    \frac{q}{2\pi_t(\theta)}+\frac{2\pi_t(\theta)-1}{2\pi_t(\theta)}&\text{if }q\geq1-\pi_t(\theta).
\end{cases}
$$

For any $\theta>\theta'$, we have$$
F_t\big(F_t^{-1}(q\vert\theta')\vert\theta\big)=\begin{cases}
    q\frac{1-\pi_t(\theta)}{1-\pi_t(\theta')}&\text{if }q<1-\pi_t(\theta'),\\
    q\frac{\pi_t(\theta)}{\pi_t(\theta')}+1-\frac{\pi_t(\theta)}{\pi_t(\theta')}&\text{if }q\geq1-\pi_t(\theta').
\end{cases}
$$
Given that $F_t\big(F_t^{-1}(q\vert\theta')\vert\theta\big)=F_{t'}\big(F_{t'}^{-1}(q\vert\theta')\vert\theta\big)$ for $q=0,1$, to have $F_t\big(F_t^{-1}(q\vert\theta')\vert\theta\big)\leq F_{t'}\big(F_{t'}^{-1}(q\vert\theta')\vert\theta\big)$ for all $q$, we must have$$
\frac{\pi_t(\theta)}{\pi_t(\theta')}\geq \frac{\pi_{t'}(\theta)}{\pi_{t'}(\theta')} \text{ and } \frac{1-\pi_{t'}(\theta)}{1-\pi_{t'}(\theta')}\geq \frac{1-\pi_t(\theta)}{1-\pi_t(\theta')}.
$$
\end{proof}

\subsection{Proof of \cref{lemma_diff}}

\begin{proof}
    Suppose there is $\theta'\in(\underline{\theta},\overline{\theta})$ such that $\pi_{t'}(\theta')<\pi_t(\theta')$. If $t\succeq_d t'$, then for all $\theta>\theta'$, $$
    \pi_t(\theta)\pi_{t'}(\theta')\geq \pi_t(\theta')\pi_{t'}(\theta).
    $$Adding $\pi_{t'}(\theta)\pi_{t'}(\theta')$ on both sides, we obtain\begin{equation}\label{cond_lemma}
        \pi_{t'}(\theta')(\pi_t(\theta)-\pi_{t'}(\theta))> \pi_{t'}(\theta)(\pi_t(\theta')-\pi_{t'}(\theta'))\Leftrightarrow \frac{\pi_t(\theta)-\pi_{t'}(\theta)}{\pi_t(\theta')-\pi_{t'}(\theta')}> \frac{\pi_{t'}(\theta)}{\pi_{t'}(\theta')},
    \end{equation}
    where we have used that $\pi_t(\theta')-\pi_{t'}(\theta')>0$.

    Test $t$ more difficult than $t'$ also implies$$
(1-\pi_t(\theta))(1-\pi_{t'}(\theta'))\geq (1-\pi_t(\theta'))(1-\pi_{t'}(\theta)).
    $$Rearranging and adding $\pi_{t'}(\theta)\pi_{t'}(\theta')$ on both sides again, we obtain $$
    (1-\pi_{t'}(\theta))(\pi_t(\theta')-\pi_{t'}(\theta'))> (1-\pi_{t'}(\theta'))(\pi_t(\theta)-\pi_{t'}(\theta))\Leftrightarrow \frac{1-\pi_{t'}(\theta)}{1-\pi_{t'}(\theta')}>\frac{\pi_t(\theta)-\pi_{t'}(\theta)}{\pi_t(\theta')-\pi_{t'}(\theta')}.
    $$Together with inequality (\ref{cond_lemma}), we can get $$
    \frac{1-\pi_{t'}(\theta)}{1-\pi_{t'}(\theta')}> \frac{\pi_t(\theta)-\pi_{t'}(\theta)}{\pi_t(\theta')-\pi_{t'}(\theta')}> \frac{\pi_{t'}(\theta)}{\pi_{t'}(\theta')}.
    $$ This implies $\pi_{t'}(\theta')> \pi_{t'}(\theta)$, a contradiction. 
\end{proof}

\subsection{Proof of \cref{prop_bertr}}

\begin{proof}
    In any equilibrium, profits must be weakly positive for otherwise the firm can just set $\alpha(x)=0$ for $x=h,l$ and increase profits.
    
    Suppose first that $\E[\theta]\geq 0$ and suppose that $\alpha(h)<1$. In any symmetric equilibrium $s=(t,\alpha)$, all agents choose either firm with probability $1/2$ and we have $$
    v(s,s,\phi)=\frac{1}{2}\int_{\Theta}\theta\Big(\pi_t(\theta)\alpha(h)+(1-\pi_t(\theta))\alpha(l)\Big)dF>0.
    $$ If one first sets $s'=(t,\alpha')$ with $\alpha'(h)=\alpha(h)+\epsilon$ and leaves the test unchanged, almost all types prefer $s'$ to $s$. The resulting profits are $$
    \int_{\Theta}\theta\Big(\pi_t(\theta)(\alpha(h)+\epsilon)+(1-\pi_t(\theta))\alpha(l)\Big)dF>\frac{1}{2}\int_{\Theta}\theta\Big(\pi_t(\theta)\alpha(h)+(1-\pi_t(\theta))\alpha(l)\Big)dF,
    $$for any $\epsilon>0$. Therefore $\alpha(h)=1$.
    
    If $\alpha(l)<1$, we can set $\alpha'(l)=\alpha(l)+\epsilon$ and leave the test unchanged. Almost all types prefer $s'$ to $s$. The resulting profits are $$
    \int_{\Theta}\theta\Big(\pi_t(\theta)+(1-\pi_t(\theta))(\alpha(l)+\epsilon)\Big)dF>\frac{1}{2}\int_{\Theta} \theta\Big(\pi_t(\theta)+(1-\pi_t(\theta))\alpha(l)\Big)dF,
    $$for $\epsilon$ small enough. We thus get that equilibrium profits are $\frac{1}{2}\E[\theta]$.

    If $\E[\theta]<0$, the same argument holds: as long as profits are strictly positive any firm can increase $\alpha(x)$ and have a strictly profitable deviation. As long as $\int_\Theta \theta\pi_t(\theta)dF>0$, there will also always be an incentive to increase $\alpha(h)$. If $\int_\Theta \theta\pi_t(\theta)dF=0$ then we must have $\alpha(l)=0$ and we could have $\alpha(h)\in[0,1]$ in equilibrium. 
\end{proof}

\subsection{Proof of \cref{theo:characterisation}}

The proof of \cref{theo:characterisation} is in three steps. First, I show that in any equilibrium, the accuracy level must be $\sigma^*$. Second, I show that any equilibrium must have difficulty level equal to $\min\{d\in D:\int_{\Theta}\theta\pi_{\sigma^*,d}(\theta)dF\geq 0\}$. Finally, I show that the strategies described in \cref{theo:characterisation} are indeed an equilibrium.  

I start with the following preliminary lemma on the acceptance rule.

\begin{lemma}\label{lemma:cutoff}
    In any symmetric equilibrium, $s=(t,\alpha)$, the acceptance rule takes a cutoff form: $\alpha(l)>0\Rightarrow\alpha(h)=1$ and $\alpha(h)<1\Rightarrow \alpha(l)=0$.

    Moreover, it is without loss of optimality to consider deviations from symmetric strategy profiles where the acceptance rule takes a cutoff form.
\end{lemma}

\begin{proof}
    Take any selection procedure. Let $\alpha^*$ be the cutoff strategy that solves$$
\alpha(l)\pi_t(l\vert0)+\alpha(h)\pi_t(h\vert0)=\alpha^*(l)\pi_t(l\vert0)+\alpha^*(h)\pi_t(h\vert0).
    $$Such strategy always exists. It is easy to verify that\begin{align*}
        &\text{for }\theta<0,\;\alpha(l)\pi_t(l\vert\theta)+\alpha(h)\pi_t(h\vert\theta)\geq\alpha^*(l)\pi_t(l\vert\theta)+\alpha^*(h)\pi_t(h\vert\theta),\\
        &\text{for }\theta>0,\;\alpha(l)\pi_t(l\vert\theta)+\alpha(h)\pi_t(h\vert\theta)\leq\alpha^*(l)\pi_t(l\vert\theta)+\alpha^*(h)\pi_t(h\vert\theta).
    \end{align*}
    These inequalities are strict if $\alpha$ is not a cutoff strategy and $\pi_t(h\vert\theta)\neq \pi_t(h\vert 0)$. 

    Therefore, all symmetric equilibria must have a cutoff acceptance rule, otherwise the cutoff rule as defined above would constitute a profitable deviation. Moreover, when testing for profitable deviations, the cutoff rule as defined here always makes the deviator at least as well off as the original rule. 
\end{proof}

\begin{lemma}\label{lemma:lehmann}
    In any symmetric equilibrium, the accuracy level is $\sigma^*=\max\Sigma$.
\end{lemma}

\begin{proof}
The proof of this lemma uses the original notion of accuracy from \citet{lehmann1988} as defined in \cref{app_prelim}. Using this definition shortens the argument and shows that its logic does not depend on the binary signal setting.

    For each test $t$, let $$
\Tilde{F}_t(x\vert\theta)=\begin{cases}
    2(1-\pi_t(\theta))x&\text{if }x\in[0,1/2),\\
    1+2\pi_t(\theta)(x-1)&\text{if }x\in[1/2,1].
\end{cases}
$$
First note that for each cutoff strategy $\alpha$ and test $t$, there is a corresponding cutoff $x \in [0,1]$ such that $\alpha(h)\pi_t(\theta)+\alpha(l)(1-\pi_t(\theta))=1-\Tilde{F}_t(x\vert\theta)$. Because each test is monotonic in types, i.e., it has the monotone likelihood ratio property, we have $\Tilde{F}_t(x\vert\theta)\leq \Tilde{F}_t(x\vert\theta')$ for any $\theta'<\theta$ and $x\in [0,1]$.

    Suppose there is a symmetric equilibrium $s'=(\tau(\sigma,d),\alpha)$ with $\sigma<\sigma^*$. Let $t'=\tau(\sigma,d)$ and $x'$ be the corresponding cutoff in the modified test $\Tilde{F}_{t'}$.

    From \cref{prop_bertr}, it must be that firms' profits are zero. Take test $t=\tau(\sigma^*,d)$ and for $\theta=0$, find the cutoff strategy $\alpha^*_\theta$ such that $$
    \alpha(h)\pi_{t'}(\theta)+\alpha(l)(1-\pi_{t'}(\theta))=\alpha^*_\theta(h)\pi_t(\theta)+\alpha^*_\theta(l)(1-\pi_t(\theta)).
    $$  and let $x^*(\theta)$ be the corresponding cutoff in the modified test. Observe that $$
    1-\Tilde{F}_t(x^*(0)\vert \theta)\geq 1-\Tilde{F}_t(x^*(\theta)\vert \theta)=1-\Tilde{F}_{t'}(x'\vert \theta),\text{ for all }\theta\in (0,\overline{\theta}),\footnote{Here, I use \citepos{lehmann1988} original definition, see \cref{app_prelim}.}
    $$where the inequality comes from the definition of accuracy (in \cref{app_prelim}) and the equality from the definition of $x^*(\theta)$. The inequality is reversed for $\theta<0$. Moreover, there must be some types for whom the inequality is strict as ${t'}\prec_a t$. This implies that if one of the firms deviates to $s=(t,\alpha^*_0)$, it achieves strictly positive profits, making them better off than in the candidate equilibrium. To see this formally, let $\Theta_s, \Theta_i,\Theta_{s'}$ be the set of types strictly preferring $s$, indifferent and strictly preferring $s'$. The profits following the deviation to $s$ are 
    
    \begin{align*}
        &\frac{1}{2}\int_{\Theta_i}(1-\Tilde{F}_t(x^*(0)\vert \theta))\theta dF+\int_{\Theta_s}(1-\Tilde{F}_t(x^*(0)\vert \theta))\theta dF\\
        &>\frac{1}{2}\int_{\Theta_i}(1-\Tilde{F}_{t'}(x^*(0)\vert \theta))\theta dF+\int_{\Theta_s}(1-\Tilde{F}_{t'}(x^*(0)\vert \theta))\theta dF\\
        &>\frac{1}{2}\int_{\Theta_i}(1-\Tilde{F}_{t'}(x^*(0)\vert \theta))\theta dF+\int_{\Theta_s}(1-\Tilde{F}_{t'}(x^*(0)\vert \theta))\theta dF\\
        &\quad+\frac{1}{2}\int_{\Theta_{s'}}(1-\Tilde{F}_{t'}(x^*(0)\vert \theta))\theta dF-\frac{1}{2}\int_{\Theta_{s}}(1-\Tilde{F}_{t'}(x^*(0)\vert \theta))\theta dF=0,
    \end{align*}
where the first inequality holds because only positive types strictly prefer $s$ and the second by adding negative terms.\end{proof}
    \begin{lemma}\label{lemma:difficulty}
         If $s=(\tau(\sigma,d),\alpha)$ is a symmetric equilibrium, then $d=\min\{d'\in D: \int_{\Theta}\theta\pi_{\sigma,d'}(\theta)dF\geq 0\}$.
    \end{lemma}

    \begin{proof}
        Suppose there is a symmetric equilibrium $s'=(\tau(\sigma,d'),\alpha')$ with $d'>d=\min\{d'\in D: \int_{\Theta}\theta\pi_{\sigma,d'}(\theta)dF\}$. Take $\Tilde{d}\in (d,d')$ and let $t=\tau(\sigma,\Tilde{d})$ and $t'=\tau(\sigma,d')$. By \cref{assump:technical}, $t$ and $t'$ are not comparable in accuracy and therefore $\frac{\pi_{t}(\theta)}{\pi_{t'}(\theta)}$ and $\frac{\pi_{t}(\theta)}{\pi_{\sigma,d}(\theta)}$ are not constant and thus $\E[\theta\vert t',h]> \E[\theta\vert t,h]>\E[\theta\vert \tau(\sigma,d),h]\geq  0$. Consider a deviation to $s=(t,\alpha)$ with $\alpha(h)=1$ and $\alpha(l)\geq 0$. 

        By \cref{prop_bertr}, profits are zero and $\alpha'(h)=1$ and $\alpha'(l)>0$. To simplify notation, I write $\alpha'$ for $\alpha'(l)$ and $\alpha$ for $\alpha(l)$. 

        Type $\theta$ chooses selection procedure $s$ if\begin{equation}\label{eq:selec_dif1}
            \pi_t(\theta)+\alpha(1-\pi_t(\theta))\geq \pi_{t'}(\theta)+\alpha'(1-\pi_{t'}(\theta)) \Leftrightarrow (1-\alpha)(1-\pi_{t}(\theta))\leq (1-\alpha')(1-\pi_{t'}(\theta)).
        \end{equation}
        
        Because $\frac{1-\pi_{t'}(\theta)}{1-\pi_{t}(\theta)}$ is increasing, if the inequality is satisfied for $\theta$, it is satisfied for all $\theta'>\theta$, i.e., there is positive selection into $s$.

        We now show that there is an $\alpha$ such that $s$ is a profitable deviation. 

        Consider first the case where $\alpha=0$. If inequality (\ref{eq:selec_dif1}) is satisfied for some types, then by the positive selection and $\E[\theta\vert t,h]>0$, $s$ is a profitable deviation. 

        If inequality (\ref{eq:selec_dif1}) is not satisfied for all types, then in particular, it is not satisfied for $\theta=0$, i.e., $$
         (1-\alpha)(1-\pi_{t}(0))>(1-\alpha')(1-\pi_{t'}(0)), \text{at }\alpha=0.
        $$At the same time, inequality (\ref{eq:selec_dif1}) holds for $\theta=0$ when $\alpha=\alpha'(l)$ as $\pi_t(\theta)\geq \pi_{t'}(\theta)$ (\cref{lemma_diff}). Therefore, by the intermediate value theorem, there is $\alpha\in (0,\alpha'(l))$ such that $$
        (1-\alpha)(1-\pi_{t}(0))=(1-\alpha')(1-\pi_{t'}(0)),
        $$ By the positive selection, $s$ is a profitable deviation.        
    \end{proof}

    \begin{lemma}\label{lemma:existence}
        The strategy $s=(\tau(\sigma^*,d^*),\alpha)$ with $\sigma^*=\max \{\sigma\in\Sigma\}$ and 
        $d^*=\min\{d\in D: \int_{\Theta}\theta\pi_{\sigma^*,d}(\theta)dF\geq 0\}$ is a symmetric equilibrium.
    \end{lemma}

    \begin{proof}
        Take a symmetric candidate equilibrium $s=(\tau(\sigma^*,d^*),\alpha)$ and let $t=\tau(\sigma^*,d^*)$. First, observe that $\alpha(h)=1$ in equilibrium. If $\int \theta \pi_{\sigma^*,d^*}(\theta)dF>0$, this follows from \cref{lemma:cutoff} and the fact that payoff must be zero in equilibrium. If $\int \theta \pi_{\sigma^*,d^*}(\theta)dF=0$, then we could have $\alpha(h)\in [0,1]$ and $\alpha(l)=0$. But then there would be a deviation to some test $\tau(\sigma^*,d')$ with $d'>d$ with $\alpha'(h)=1$, $\alpha'(l)=0$. To see this, observe that $\pi_{\sigma^*,d^*}(\theta)>\alpha(h)\pi_{\sigma^*,d^*}(\theta)$ for all $\theta\in (\underline{\theta},\overline{\theta})$ and therefore by continuity in $d$, there is $d'$ with $\pi_{\sigma^*,d'}(\theta)>\alpha(h)\pi_{\sigma^*,d^*}(\theta)$ for some $\theta$. Moreover, there is positive selection into $\tau(\sigma^*,d')$ as $\frac{\pi_{\sigma^*,d'}(\theta)}{\pi_{\sigma^*,d^*}(\theta)}$ is increasing. Combined with the fact that $\int_\Theta \theta\pi_{\sigma^*,d'}(\theta)dF>0$ by \cref{prop_charac_dif}, $s'=(\tau(\sigma^*,d'),\alpha')$ is a profitable deviation.
        
        Set $\alpha(l)\geq 0$ and $\alpha(h)=1$. To simplify notation, write $\alpha(l)=\alpha$. 

        Take a deviation to $t'=\tau(\sigma,d)$ with $\sigma\leq \sigma^*$.

        First, suppose that selection procedure $(\tau(\sigma,d),\alpha')$ with $\sigma<\sigma^*$ is a profitable deviation. In that case, $(\tau(\sigma^*,d),\alpha'')$ where $\alpha''$ is such that positive types have a higher probability of being accepted and negative types a lower one is also a profitable deviation. Such $\alpha''$ exists using the same reasoning as in the proof of \cref{lemma:difficulty}.

        Take test $t'=\tau(\sigma^*,d)$ with $d>d^*$. We have $\tau(\sigma^*,d)\succeq_d \tau(\sigma^*,d^*)$ and therefore for all $\theta,\theta'\in (\underline{\theta},\overline{\theta})$ with $\theta>\theta'$,$$
        \frac{1-\pi_{\sigma^*,d^*}(\theta)}{1-\pi_{\sigma^*,d^*}(\theta')}\leq \frac{1-\pi_{\sigma^*,d}(\theta)}{1-\pi_{\sigma^*,d}(\theta')}. 
        $$

        Note that any deviation to $s'=(\tau(\sigma^*,d),\alpha')$ must have $\alpha'(l)>\alpha(l)\geq 0$ as $\pi_{\sigma^*,d^*}(\theta)\geq \pi_{\sigma^*,d}(\theta)$ for all $\theta$. Otherwise, no type would choose $s'$.

        Let $\Delta u(\theta)$ denote the difference in probability of being accepted of type $\theta$ between $s'$ and $s$:\begin{align*}
             \Delta u(\theta)&=\pi_{t'}(\theta)+\alpha'(1-\pi_{t'}(\theta))-\pi_t(\theta)-\alpha(1-\pi_t(\theta))\\
             &=(1-\pi_{t'}(\theta))(\alpha'-1+(1-\alpha)\frac{1-\pi_t(\theta)}{1-\pi_{t'}(\theta)}).
        \end{align*}
        Both $1-\pi_{t'}(\theta)$ and $\alpha'-1+(1-\alpha)\frac{1-\pi_t(\theta)}{1-\pi_{t'}(\theta)}$ are decreasing functions of $\theta$. Therefore, whenever $\alpha'-1+(1-\alpha)\frac{1-\pi_t(\theta)}{1-\pi_{t'}(\theta)}\geq 0$, and thus $\Delta u(\theta)\geq 0$, $\Delta u(\theta)$ is decreasing as the product of two positive decreasing functions is decreasing. Moreover, the function $\Delta u(\theta)$ is single crossing from above.

        Let $\Theta_{s'},\Theta_i$ and $\Theta_s$ be the set of types for which $\Delta u(\theta)>0$, $=0$ and $<0$.

        The profits from $s'$ are \begin{align*}
            &\int_{\Theta_{s'}}(\pi_{t'}(\theta)+\alpha'(1-\pi_{t'}(\theta))\theta dF+\frac{1}{2}\int_{\Theta_i}(\pi_{t'}(\theta)+\alpha'(1-\pi_{t'}(\theta))\theta dF
            \\&\quad-\int_{\Theta}(\pi_{t}(\theta)+\alpha(1-\pi_{t}(\theta))\theta dF\\
            &=\int_{\Theta_{s'}}\theta\Delta u(\theta)dF+\frac{1}{2}\int_{\Theta_{i}}\theta\Delta u(\theta)dF\\
            &-\frac{1}{2}\left(\int_{\Theta_i}(\pi_{t}(\theta)+\alpha(1-\pi_{t}(\theta))\theta dF+\int_{\Theta_s}(\pi_{t}(\theta)+\alpha(1-\pi_{t}(\theta))\theta dF\right)\\
            &-\frac{1}{2}\int_{\Theta_s}(\pi_{t}(\theta)+\alpha'(1-\pi_{t}(\theta))\theta dF,
        \end{align*}
        where we have added the profits $s$ which are zero on the first line and rearranged to obtain the inequality. Now note that if $s'$ is a profitable deviation, it must be that $\theta=0\in \Theta_{s'}$ or $\in \Theta_i$, i.e., some positive types must be choosing $s'$. This implies that all types in $\Theta_s$ must be positive. We can now prove that the profits from $s'$ are negative. 

        The last term is negative as all types in $\Theta_s$ are positive. The second last term is negative as it is the profits from equilibrium after having removed types at the bottom. The term $\int_{\Theta_i}\theta\Delta u(\theta)dF=0$ by definition of $\Theta_i$. Finally, because on $\Theta_{s'}$, $\Delta u(\theta)$ is decreasing, we have$$\int_{\Theta_{s'}}\theta\Delta u(\theta)dF\leq \int_{\Theta_{s'}}\theta\Delta u(0)dF\leq 0,
        $$using that $\E[\theta]<0$.

        Now consider a deviation to $s'=(\tau(\sigma^*,d),\alpha')$ with $d<d^*$. If such test is feasible, it must be that $\int_\Theta \pi_{\sigma^*,d^*}(\theta)\theta dF=0$ (by continuity of $\pi_{\sigma,d}$) and $\int_\Theta \pi_{\sigma^*,d}(\theta)\theta dF<0$.

        The candidate equilibrium $s=(\tau(\sigma^*,d^*),\alpha)$ has $\alpha(l)=0$

        Let $t'=\tau(\sigma^*,d)$. If $\alpha'(l)> 0$, then all types choose $s'$ as $\pi_{d,\sigma^*}(\theta)\geq \pi_{d^*,\sigma^*}$. But $\E[\theta\vert t',h]<0$ and therefore payoffs from $s'$ are negative. Therefore $\alpha'(l)=0$.

        Let $\Delta u(\theta)$ the difference in probability of acceptance between $s'$ and $s$ for type $\theta$:$$
        \Delta u(\theta)=\alpha'(h)\pi_{t'}(\theta)-\pi_t(\theta)=\pi_{t'}(\theta)(\alpha'(h)-\frac{\pi_t(\theta)}{\pi_{t'}(\theta)}).
        $$The function $\Delta u(\theta)$ is single-crossing from above and therefore there is negative selection into $s'$. Combined with the fact that $\E[\theta\vert t',h]<0$, the deviation is not profitable.
    \end{proof}

\subsection{Proof of \cref{prop:comparative_FOSD}}

Let $(\alpha,\sigma,d)$ and $(\alpha',\sigma',d')$ be the equilibrium strategy under $(F,\Sigma\times D)$ and $(F',\Sigma\times D)$. 

Using standard arguments, the LR ordering implies that for any $\alpha(l)$,$$
\frac{\int\theta(\pi_{\sigma,d}(\theta)+\alpha(l)(1-\pi_{\sigma,d}(\theta)))dF}{\int(\pi_{\sigma,d}(\theta)+\alpha(l)(1-\pi_{\sigma,d}(\theta)))dF}\leq \frac{\int\theta(\pi_{\sigma,d}(\theta)+\alpha(l)(1-\pi_{\sigma,d}(\theta)))dF'}{\int(\pi_{\sigma,d}(\theta)+\alpha(l)(1-\pi_{\sigma,d}(\theta)))dF'}.
$$

If $d=\min D=\min\{\Tilde{d}:\int\theta\pi_{\sigma,\Tilde{d}}(\theta)dF\geq 0\}$, then, by the likelihood ratio (LR) ordering, $d'=\min D=\min\{\Tilde{d}:\int\theta\pi_{\sigma,\Tilde{d}}(\theta)dF'\geq 0\}$. By \cref{theo:characterisation}, we also have $\sigma=\sigma'$.

Because $$
0=\frac{\int\theta(\pi_{\sigma,d}(\theta)+\alpha(l)(1-\pi_{\sigma,d}(\theta)))dF}{\int(\pi_{\sigma,d}(\theta)+\alpha(l)(1-\pi_{\sigma,d}(\theta)))dF}\leq \frac{\int\theta(\pi_{\sigma,d}(\theta)+\alpha(l)(1-\pi_{\sigma,d}(\theta)))dF'}{\int(\pi_{\sigma,d}(\theta)+\alpha(l)(1-\pi_{\sigma,d}(\theta)))dF'},
$$
and $$
0=\int\theta(\pi_{\sigma,d}(\theta)+\alpha'(l)(1-\pi_{\sigma,d}(\theta))dF',
$$ 
we must have $\alpha(h)=\alpha'(h)=1$ and $\alpha(l)\leq \alpha'(l)$. Therefore, $p^*(\theta;F',\Sigma\times D)\geq p^*(\theta;F,\Sigma\times D)$.

Suppose that $d=\min\{\Tilde{d}:\int\theta\pi_{\sigma,\Tilde{d}}(\theta)dF\geq 0\}>\min D$. Then by the LR ordering, $d'=\min\{\Tilde{d}:\int\theta\pi_{\sigma,\Tilde{d}}(\theta)dF'\geq 0\}\leq d$. Again by \cref{theo:characterisation}, we have $\sigma'=\sigma$.

We also have $\int\theta\pi_{\sigma,d}(\theta)dF=0$ and therefore $\alpha(l)=0$ and $\alpha(h)=1$. Furthermore, for all types, $\pi_{\sigma,d}(\theta)\leq \pi_{\sigma,d'}(\theta)$ and therefore, $p^*(\theta;F',\Sigma\times D)\geq p^*(\theta;F,\Sigma\times D)$.

\subsection{Proof of \cref{prop:comparative_tech}}

Let $(\alpha,\sigma,d)$ and $(\alpha',\sigma',d')$ be the equilibrium strategy under $(F,\Sigma\times D)$ and $(F,\Sigma'\times D')$. 

By \cref{theo:characterisation}, we have $\sigma'\geq \sigma$ and $d'\leq d$. 

Therefore, we have for all $\theta'>\theta$, \begin{align*}
&\frac{(1-\alpha(l))(1-\pi_{\sigma,d}(\theta'))}{(1-\alpha(l))(1-\pi_{\sigma,d}(\theta))} \geq
\frac{(1-\alpha'(l))(1-\pi_{\sigma',d'}(\theta'))}{(1-\alpha'(l))(1-\pi_{\sigma',d'}(\theta))}\\
\Leftrightarrow \;&\frac{1-\pi_{\sigma,d}(\theta')-\alpha(l)(1-\pi_{\sigma,d}(\theta'))}{1-\pi_{\sigma,d}(\theta)-\alpha(l)(1-\pi_{\sigma,d}(\theta))}\geq \frac{1-\pi_{\sigma',d'}(\theta')-\alpha'(l)(1-\pi_{\sigma',d'}(\theta'))}{1-\pi_{\sigma',d'}(\theta)-\alpha'(l)(1-\pi_{\sigma',d'}(\theta))}.    
\end{align*}

Define the tests $\pi_{t}(\theta)=\pi_{\sigma,d}(\theta)+\alpha(l)(1-\pi_{\sigma,d}(\theta))$ and $\pi_{t'}(\theta)=\pi_{\sigma',d'}(\theta)+\alpha'(l)(1-\pi_{\sigma',d'}(\theta))$. By \cref{prop_charac_dif}, we have $\E[\theta\vert t',l]\leq \E[\theta\vert t,l]$. 

At the same time, by \cref{prop_bertr}, the zero profits condition holds in equilibrium and$$
\int\theta(\pi_{\sigma,d}(\theta)+\alpha(l)(1-\pi_{\sigma,d}(\theta)))dF=\int\theta(\pi_{\sigma',d'}(\theta)+\alpha'(l)(1-\pi_{\sigma',d'}(\theta)))dF,
$$and therefore$$
\int \theta(1-\alpha(l))(1-\pi_{\sigma,d}(\theta))dF=\int \theta(1-\alpha'(l))(1-\pi_{\sigma',d'}(\theta))dF.
$$
We get\begin{align*}
    &\frac{\int \theta(1-\alpha(l))(1-\pi_{\sigma,d}(\theta))dF}{\int (1-\alpha(l))(1-\pi_{\sigma,d}(\theta))dF}\geq\frac{\int \theta(1-\alpha(l))(1-\pi_{\sigma',d'}(\theta))dF}{\int (1-\alpha(l))(1-\pi_{\sigma',d'}(\theta))dF}\\
    \Leftrightarrow\;&\int (1-\alpha(l))(1-\pi_{\sigma,d}(\theta))dF\geq \int (1-\alpha(l))(1-\pi_{\sigma',d'}(\theta))dF\\
    \Leftrightarrow \;&\int (\pi_{\sigma,d}(\theta)+\alpha(l)(1-\pi_{\sigma,d}(\theta)))dF\leq \int (\pi_{\sigma',d'}(\theta)+\alpha(l)(1-\pi_{\sigma',d'}(\theta)))dF, 
\end{align*}
where on the second line we use the fact that $\E[\theta\vert t',l]\leq \E[\theta\vert t,l]<0$. Therefore, $\E[p^*(\theta;F,\Sigma\times D)]\leq \E[p^*(\theta;F,\Sigma'\times D')]$.

The single-crossing condition can be established as follows\begin{align*}
    &\pi_{t'}(\theta)-\pi_t(\theta)\geq 0\\
    \Leftrightarrow\;&1-\pi_t(\theta)-(1-\pi_{t'}(\theta))\geq 0\\
    \Leftrightarrow\; &\frac{1-\alpha(l)}{1-\alpha'(l)}\frac{1-\pi_{\sigma,d}(\theta)}{1-\pi_{\sigma',d'}(\theta)}-1\geq 0.
\end{align*}
The expression on the LHS is increasing in $\theta$ and therefore the single-crossing condition holds.

\subsection{Proof of \cref{prop_capacity}}

\begin{proof}
    We show that the strategy $s=(\tau(\overline{d}),\alpha)$ with $\alpha(h)=1$, $\alpha(l)=0$ is an equilibrium. For simplicity, let $t=\tau(\overline{d})$. Equilibrium payoffs for both firms are $$
    k\cdot \E[\theta\vert t,h],
    $$
    given the assumption on the capacity constraints and that by \cref{lemma_diff}, $\pi_{\overline{d}}(\theta)\leq \pi_d(\theta)$ for all $\theta\in \Theta$ and $d\in D$.
    
    Consider a deviation of firm 1 to $s'=(t'=\tau(d'),\alpha')$ with $d'<\overline{d}$. Let $p_i$ denote the probability that a given type has its application considered by firm $i$. Suppose first that $\alpha'(l)=0$. For simplicity, let $\alpha'(h)=\alpha'$. I will first show that selection into firm 1 is negative. The agent's utility difference between firm 1 and 2 is $$
    \Delta u(\theta)=p_1\alpha'\pi_{t'}(\theta)-p_2\pi_t(\theta).
    $$
    Using that $\frac{\pi_t(\theta)}{\pi_{t'}(\theta)}$ is increasing in $\theta$, we can establish that selection is negative. Let $\Theta_{t'}$ and $\Theta_i$ be the sets of types for which $\Delta u(\theta)>0$ and $\Delta u(\theta)=0$. If the capacity constraint is binding, then $$
    k\cdot \E[\theta\vert t',h,\theta\text{ chooses 1}]\leq k\cdot\E[\theta\vert t',h]\leq k\cdot\E[\theta\vert t,h].
    $$ If the capacity constraint is not binding, i.e., $k\geq p_1$, the payoffs from deviating are\begin{align*}
    \int_{\Theta_{t'}}\theta \alpha'\pi_{t'}(\theta)dF+\frac{1}{2}\int_{\Theta_i}\theta\alpha'\pi_{t'}(\theta)dF&\leq \frac{k}{p_1}\big[\int_{\Theta_{t'}}\theta \alpha'\pi_{t'}(\theta)dF+\frac{1}{2}\int_{\Theta_i}\theta\alpha'\pi_{t'}(\theta)dF\big]\\
    &=k\cdot \E[\theta\vert t',h,\theta\text{ chooses 1}]\\
    &\leq k\cdot\E[\theta\vert t',h]\leq k\cdot\E[\theta\vert t,h],
    \end{align*}
where we use that $k\geq p_1$ on the first line, the negative selection and \cref{prop_charac_dif} on the third line.

    Now suppose that $\alpha'(l)>0$. For simplicity let $\alpha'(l)=\alpha'$. First observe that to have a profitable deviation, it must be $p_2/p_1>1$, otherwise no type would choose firm 2 as $\pi_{t'}(\theta)\geq \pi_t(\theta)$. But in that case firm 1's capacity constraint binds and its profits are lower than under the equilibrium profits as $\E[\theta\vert t,h]\geq \E[\theta\vert t',h]$. Let $p=p_2/p_1$ and $$
    \Delta u(\theta)=p_1(\pi_{t'}(\theta)+\alpha'(1-\pi_{t'}(\theta)))-p_2\pi_t(\theta)=\pi_t(\theta)\left(p_1(1-\alpha')\frac{\pi_{t'}(\theta)}{\pi_t(\theta)}+\frac{p_1\alpha'}{\pi_t(\theta)}-p_2\right).
    $$
    The term in brackets is negative whenever $\Delta u\leq 0$ and it is decreasing while $\pi_t(\theta)$ is increasing and positive. Therefore, whenever $\Delta u(\theta)\leq 0$, $\Delta u(\theta)$ is decreasing. This shows there is negative selection into $s'$ and the acceptance probability satisfies decreasing differences for the types choosing $s'$. By a similar argument as above, the deviation cannot be profitable. 
\end{proof}

\subsection{Proof of \cref{prop_two}}

I show that there is an equilibrium where firm 1 chooses selection procedure $(\overline{t},\alpha(h)=1, \alpha(l)=0)$ and firm 2 chooses $(t,\alpha(h)=1, \alpha(l)=0)$. To simplify notation, I denote by $\overline{\pi}_i$ and $\underline{\pi}_i$ the probability of $\overline{\theta}$ and $\underline{\theta}$ generate the high signal in the test chosen by firm $i$.

In the suggested equilibrium, type $\overline{\theta}$ chooses firm 1 and type $\underline{\theta}$ mixes between firm 1 and firm 2. Denote by $\phi$ the probability that type $\underline{\theta}$ chooses firm 2. Both firms' capacities are binding. 

Type $\underline{\theta}$ is willing to mix between firm 1 and firm 2 if$$
\frac{k}{(1-\mu)\phi\underline{\pi}_2}\underline{\pi}_2=\frac{k}{\mu\overline{\pi}_1+(1-\mu)(1-\phi)\underline{\pi}_1}\underline{\pi}_1.
$$
Solving for $\phi$, we get $\phi=\frac{\mu\overline{\pi}_1+(1-\mu)\underline{\pi}_1}{2(1-\mu)\underline{\pi}_1}$. We have $\phi\leq 1\Leftrightarrow \mu\overline{\pi}_1\leq (1-\mu)\underline{\pi}_1$. This inequality corresponds to the second condition in \cref{prop_two}. Note also that we have $\phi\geq \frac{1}{2(1-\mu)}$.

Type $\overline{\theta}$ prefers firm 1 if \begin{equation}\label{eq:IC_high}
    \frac{k}{\mu\overline{\pi}_1+(1-\mu)(1-\phi)\underline{\pi}_1}\overline{\pi}_1\geq \frac{k}{(1-\mu)\phi\underline{\pi}_2}\overline{\pi}_2 \Leftrightarrow \frac{\overline{\pi}_1}{\underline{\pi}_1}\geq \frac{\overline{\pi}_2}{\underline{\pi}_2}.
\end{equation}
This inequality is always satisfied by definition of $\overline{t}$, the test chosen by firm 1.

The capacity constraints bind if \begin{align*}
    &\text{Firm 1: }\mu\overline{\pi}_1+(1-\mu)(1-\phi)\underline{\pi}_1\geq k,\\
    &\text{Firm 2: }(1-\mu)\phi\underline{\pi}_2\geq k.
\end{align*}
Plugging in the value of $\phi$ for the first inequality, we obtain the third condition in \cref{prop_two}. The second inequality corresponds to the fourth inequality of \cref{prop_two} using that $\phi\geq \frac{1}{2(1-\mu)}$.

Let's now consider whether firms have profitable deviations.\footnote{Note that there always exists a continuation equilibrium where each type mixes with the same probability as each candidate's utility is linear in their action and continuous in the probability each type chooses each firm.} First we check whether firm 2 has a profitable deviation. 

The first observation is that firm 2 must make positive profits in equilibrium. This is the case only if $\underline{\theta}\geq 0$ (first condition in \cref{prop_two}).

Consider a deviation to $s'=(t',\alpha')$. If firm 2 deviates, there are two possibilities. Either its capacity constraint binds following the deviation or it does not. 

If firm 2's capacity constraint still binds in the continuation equilibrium, then we still have that $\underline{\theta}$ mixes between firm 1 and 2 using the same strategy (that only depended on the strategy of firm 1). Type $\overline{\theta}$ would want to deviate to firm 2 if$$
\frac{\overline{\pi}_1}{\underline{\pi}_1}<\frac{\alpha'(h)\pi_{t'}(\overline{\theta})+\alpha'(l)(1-\pi_{t'}(\overline{\theta}))}{\alpha'(h)\pi_{t'}(\underline{\theta})+\alpha'(l)(1-\pi_{t'}(\underline{\theta}))},
$$using the same calculations as (\ref{eq:IC_high}). The inequality above is never satisfied because$$
\frac{\overline{\pi}_1}{\underline{\pi}_1}\geq \frac{\pi_{t'}(\overline{\theta})}{\pi_{t'}(\underline{\theta})}\geq\frac{\alpha'(h)\pi_{t'}(\overline{\theta})+\alpha'(l)(1-\pi_{t'}(\overline{\theta}))}{\alpha'(h)\pi_{t'}(\underline{\theta})+\alpha'(l)(1-\pi_{t'}(\underline{\theta}))},
$$ where the first inequality follows from the definition of $\overline{t}$, the test used by firm 1. The inequality above also shows that there is negative selection into firm 2's selection procedure. Therefore, firm 2's profits following such deviation are still $k\cdot \underline{\theta}$, the same as the equilibrium profits.

If firm 2's capacity constraint does not bind following the deviation, because there is negative selection into the selection procedure of firm 2, at most type $\underline{\theta}$ applies to firm 2. Indeed, by the second condition of \cref{prop_two}, $(1-\mu)\underline{\pi}_1\geq \mu\overline{\pi}_1$, we obtain $1-\mu\geq 1/2$. Therefore $(1-\mu)\underline{\pi}_2\geq \frac{\underline{\pi}_2}{2}\geq k$. But then the deviating payoff is lower than the equilibrium payoffs as $k\underline{\theta}\geq (1-\mu)Pr[\underline{\theta}\text{ chooses firm 2}]\underline{\theta}$.

We now need to verify that firm 1 does not have any profitable deviation. Again, we need to distinguish the cases where the capacity constraint is binding or not in the continuation equilibrium.

Consider a deviation to $s'=(t',\alpha')$. Let $\overline{q}=\alpha'(h)\pi_{t'}(\overline{\theta})+\alpha'(l)(1-\pi_{t'}(\overline{\theta}))$, the acceptance probability under $s'$ of $\overline{\theta}$ and define $\underline{q}$ similarly.

First, there is never a profitable deviation where $\frac{\overline{q}}{\underline{q}}<\frac{\overline{\pi}_2}{\underline{\pi}_2}$. If this is the case, there is negative selection into firm 1's selection procedure so the best-case scenario is that all types choose firm 1. The payoffs are then\begin{align*}
    &\min\{1,\frac{k}{\mu\overline{q}+(1-\mu)\underline{q}}\}(\mu\overline{q}\overline{\theta}+(1-\mu)\underline{q}\underline{\theta})\\
    &\leq \frac{k}{\mu\overline{q}+(1-\mu)\underline{q}}(\mu\overline{q}\overline{\theta}+(1-\mu)\underline{q}\underline{\theta})\\
    &\leq \frac{k}{\mu\overline{q}+(1-\mu)(1-\phi)\underline{q}}(\mu\overline{q}\overline{\theta}+(1-\mu)(1-\phi)\underline{q}\underline{\theta})\\
    &\leq \frac{k}{\mu\overline{\pi}_1+(1-\mu)(1-\phi)\underline{\pi}_1}\big(\mu\overline{\pi}_1\overline{\theta}+(1-\mu)(1-\phi)\underline{\pi}_1\underline{\theta}\big),
\end{align*}
using that $\frac{\mu}{(1-\mu)(1-\phi)}\geq \frac{\mu}{1-\mu}$ on the third line and $\frac{\overline{\pi}_1}{\underline{\pi}_1}\geq \frac{\overline{q}}{\underline{q}}$ on the fourth.

If$$
\frac{\mu \overline{q}+(1-\mu)\underline{q}}{2}\geq k.
$$Then the same equilibrium as on-path where $\underline{\theta}$ mixes holds as long as $\frac{\overline{q}}{\underline{q}}\geq\frac{\overline{\pi}_2}{\underline{\pi}_2}$. Moreover, we necessarily have that $$
\frac{\overline{\pi}_1}{\underline{\pi}_1}\geq \frac{\overline{q}}{\underline{q}}. 
$$Let $\phi_q=\frac{\mu\overline{q}+(1-\mu)\underline{q}}{2(1-\mu)\underline{q}}$ be the probability of type $\underline{\theta}$ to choose firm 2 in the continuation equilibrium. The strategy $\phi_q$ is decreasing in the likelihood ratio $\frac{\overline{q}}{\underline{q}}$ and therefore this deviation is not profitable for firm 1 as it attracts more $\underline{\theta}$.

Suppose now that \begin{equation}\label{eq:cond_k}
    \mu\overline{q}+(1-\mu)(1-\phi_q)\underline{q}=\frac{\mu \overline{q}+(1-\mu)\underline{q}}{2}< k.
\end{equation}

We cannot have both firms' capacity constraints bind as the equilibrium involving type $\underline{\theta}$ mixing does not make the capacity constraint bind. Furthermore $\mu\overline{q}<\mu\overline{q}+(1-\mu)(1-\phi_q)\underline{q}<k$, therefore, if $\underline{\theta}$ only applies to firm 1, the capacity constraint does not bind either. If all types apply to firm 1 with probability one, then the capacity constraint of firm 2 does not bind. Therefore, either only one firm's capacity constraint is binding or neither.

If only firm 1's capacity constraint binds, then if $\Tilde{\phi}$ is the probability of $\underline{\theta}$ choosing firm 2,\footnote{Again, if only type $\overline{\theta}$ is choosing firm 1, then the capacity constraint is not binding.} then$$
\mu\overline{q}+(1-\mu)(1-\phi_q)\underline{q}<k\leq\mu\overline{q}+(1-\mu)(1-\Tilde{\phi})\underline{q}.
$$This inequality implies that $(1-\Tilde{\phi})>(1-\phi_q)>(1-\phi)$ and therefore\begin{align*}
&\frac{k}{\mu\overline{\pi}_1+(1-\mu)(1-\phi)\underline{\pi}_1}\big(\mu\overline{\pi}_1\overline{\theta}+(1-\mu)(1-\phi)\underline{\pi}_1\underline{\theta}\big)\\
&\quad\quad\quad\geq\frac{k}{\mu\overline{q}+(1-\mu)(1-\phi)\underline{q}}\big(\mu\overline{q}\overline{\theta}+(1-\mu)(1-\phi)\underline{q}\underline{\theta}\big)\\
&\quad\quad\quad>\frac{k}{\mu\overline{q}+(1-\mu)(1-\Tilde{\phi})\underline{q}}\big(\mu\overline{q}\overline{\theta}+(1-\mu)(1-\Tilde{\phi})\underline{q}\underline{\theta}\big),  
\end{align*}
using that $\frac{\overline{\pi}_1}{\underline{\pi}_1}\geq \frac{\overline{q}}{\underline{q}}$ on the second line and $\frac{\mu}{(1-\mu)(1-\phi)}>\frac{\mu}{(1-\mu)(1-\Tilde{\phi})}$ on the third. 

If firm 1's capacity constraint does not bind, any continuation equilibrium where only $\overline{\theta}$ chooses firm 1 is not a profitable deviation as\begin{align*}
\overline{\phi}\mu\overline{q}\overline{\theta}&\leq \mu\overline{q}\overline{\theta}+ (1-\mu)(1-\phi_q)\underline{q}\underline{\theta}\\
    &<\frac{k}{\mu\overline{q}+(1-\mu)(1-\phi_q)\underline{q}}\big(\mu\overline{q}\overline{\theta}+ (1-\mu)(1-\phi_q)\underline{q}\underline{\theta}\big)\\
    &\leq \frac{k}{\mu\overline{\pi}_1+(1-\mu)(1-\phi)\underline{\pi}_1}\big(\mu\overline{\pi}_1\overline{\theta}+(1-\mu)(1-\phi)\underline{\pi}_1\underline{\theta}\big),
\end{align*}
where we use that $k>\mu\overline{q}+(1-\mu)(1-\phi_q)\underline{q}$ on the second line.

If all types choose firm 1 with probability one, this is not a profitable deviation either as\begin{align*}
    \mu\overline{q}\overline{\theta}+ (1-\mu)\underline{q}\underline{\theta}&\leq \frac{k}{\mu\overline{q}+ (1-\mu)(1-\phi_q)\underline{q}}\big(\mu\overline{q}\overline{\theta}+ (1-\mu)(1-\phi_q)\underline{q}\underline{\theta}\big)\\
    &\leq  \frac{k}{\mu\overline{\pi}_1+(1-\mu)(1-\phi)\underline{\pi}_1}\big(\mu\overline{\pi}_1\overline{\theta}+(1-\mu)(1-\phi)\underline{\pi}_1\underline{\theta}\big),
\end{align*}
where we use that $\frac{k}{\mu\overline{q}+ (1-\mu)(1-\phi_q)\underline{q}}\geq 1$ as the capacity constraint does not bind.

Therefore, type $\underline{\theta}$ must be mixing. If both firms' capacity constraints do not bind then $\underline{\pi}_2=\underline{q}$. But $\underline{\pi}_2/2\geq k$, a contradiction with firm 2's capacity constraint not binding.

Thus firm 2's capacity constraint must bind and, setting $\Tilde{\phi}$ as the probability $\underline{\theta}$ chooses firm 2, $$
\underline{q}=\frac{k}{(1-\mu)\Tilde{\phi}\underline{\pi}_2}\underline{\pi}_2\Leftrightarrow \Tilde{\phi}=\frac{k}{(1-\mu)\underline{q}}.
$$In equilibrium, we must have $\frac{k}{(1-\mu)\Tilde{\phi}\underline{\pi}_2}\leq 1\Leftrightarrow \underline{\pi}_2\leq \underline{q}$. Therefore, combined with the fourth condition of \cref{prop_two}, we have $\frac{\underline{q}}{2}\geq k$.

Now, we can rearrange (\ref{eq:cond_k}) to obtain\begin{align*}
    & k>\frac{\mu \overline{q}+(1-\mu)\underline{q}}{2}\geq \frac{\mu\overline{q}}{2}+(1-\mu)k\\    \Rightarrow&\;k>\frac{\overline{q}}{2}.
\end{align*}
Therefore we have $\frac{\underline{q}}{2}\geq k>\frac{\overline{q}}{2}$, a contradiction.

\subsection{Proof of \cref{prop_money}}

\begin{proof}    
    In any symmetric equilibrium, both firms get zero profits. If it is not the case and a firm gets strictly positive profits, then there is at least one signal at which firms get positive profits. Then one of them can raise the transfer by $\epsilon$ and attract all agents for an arbitrarily small increase.

    First, I show that there is no `cross-subsidisation' in equilibrium, i.e., $\alpha(l)=0$. 

    Suppose it is not the case. Let $s=(t,\alpha,m)$ be the selection procedure in equilibrium. Let $m(h)=m_h$ and $m(l)=m_l$. First note that if $m_l=0$, then any firm can decrease $\alpha(l)$ and increase its profits. Because $m_l=0$, this does not change the payoffs of the agents. So $m_l>0$. We have$$
    \int_\Theta \pi_t(\theta)\alpha_h(\theta-m_h)+(1-\pi_t(\theta))\alpha_l(\theta-m_l)dF=0.
    $$Consider the following deviation $s'$ that leaves all aspects of the selection procedure unchanged except that $m'(h)=m_h+\epsilon$ and $m'(l)=m_l-\delta$ with $\epsilon,\delta>0$ such that $$
    \int_\Theta \pi_t(\theta)\alpha_h(\theta-m_h-\epsilon)+(1-\pi_t(\theta))\alpha_l(\theta-m_l+\delta)dF=0.
    $$
    We choose $\epsilon,\delta$ small enough such that $m'(l)>0$. I will show that some types will choose the deviating firm and that the deviation will exhibit positive selection. This will imply that the deviation is profitable. 

    Type $\theta$ chooses $s'$ if \begin{gather*}
        \pi_t(\theta)\alpha_h(m_h+\epsilon)+(1-\pi_t(\theta))\alpha_l(m_l-\delta)\geq \pi_t(\theta)\alpha_hm_h+(1-\pi_t(\theta))\alpha_lm_l\\
        \Leftrightarrow\pi_t(\theta)\alpha_h\epsilon\geq (1-\pi_t(\theta))\alpha_l\delta\\
        \Leftrightarrow \frac{\pi_t(\theta)}{1-\pi_t(\theta)}\alpha_l\epsilon\geq \alpha_l\delta.
    \end{gather*}
    Because the LHS is increasing in $\theta$, there is positive selection into $s'$.
    
   Substituting for $\delta$ as a function of $\epsilon$ and the fact that the original profits are zero, we get that type $\theta$ chooses $s'$ if$$
    \pi_t(\theta)\int_\Theta (1-\pi_t(\theta))dF\leq (1-\pi_t(\theta)) \int_\Theta \pi_t(\theta)dF. 
    $$This has to hold for some types, as otherwise we have $\pi_t(\theta)\int_\Theta (1-\pi_t(\theta))dF> (1-\pi_t(\theta)) \int_\Theta \pi_t(\theta)dF$ for all $\theta$ which implies $\int_\Theta\pi_t(\theta)dF\cdot\int_\Theta (1-\pi_t(\theta))dF> \int_\Theta (1-\pi_t(\theta))dF\cdot \int_\Theta \pi_t(\theta)dF$, a contradiction. Using the same argument, some types must prefer the original selection procedure $s$. Therefore, we get positive selection into the new selection procedure. Since, absent positive selection, the profits are zero, this must be a strictly profitable deviation.
    
    Let $\overline{t}=\tau(\overline{d})$. Suppose there is a symmetric equilibrium with $t\prec_d\overline{t}$ where $t=\tau(d)$ for some $d<\overline{d}$. 

    In equilibrium, it must be that $m(l)=0$ and from the zero profits condition, $m(h)=m=\frac{\int_\Theta \theta\pi_t(\theta)dF}{\int_\Theta\pi_t(\theta)dF}$.

    Take some $\epsilon\in(0,\frac{\pi_t(0)-\pi_{\overline{t}}(0)}{\pi_{\overline{t}}(0)})$. We want to find $t'$ with $t\prec_d t'\prec_d \overline{t}$ such that $$
    m\pi_t(0)=m(1+\epsilon)\pi_{t'}(0).
    $$For $t'=t$, we have $m\pi_t(0)<m(1+\epsilon)\pi_{t'}(0)$ and for $t'=\overline{t}$ we have $m\pi_t(0)>m(1+\epsilon)\pi_{t'}(0)$, using the bound on $\epsilon$. Because $D$ is an interval and the continuity assumption (\cref{assump:technical}), by the intermediate value theorem, there is $d'$ and $t'=\tau(d')$ with $t\prec_d t'\prec_d \overline{t}$ such that $m\pi_t(0)=m(1+\epsilon)\pi_{t'}(0)$. 

    We want to show that for $\epsilon$ small enough, this constitutes a profitable deviation, i.e., $$
    \int_0^{\overline{\theta}}\pi_{t'}(\theta)(\theta-m(1+\epsilon))dF>0.
    $$Because $t'$ is more difficult than $t$, we have $\E[\theta\vert t,\theta\in[0,\overline{\theta}]]<\E[\theta\vert t',\theta\in[0,\overline{\theta}]]$. Moreover, we have that for $\epsilon$ small enough, $\E[\theta\vert t,\theta\in[0,\overline{\theta}]]>\E[\theta\vert t](1+\epsilon)$. Combining these facts, we get $$
    \frac{\int_0^{\overline{\theta}}\pi_{t'}(\theta)\theta dF}{\int_0^{\overline{\theta}}\pi_{t'}(\theta) dF}>\frac{\int_0^{\overline{\theta}}\pi_{t}(\theta)\theta dF}{\int_0^{\overline{\theta}}\pi_{t}(\theta) dF}>\frac{\int_{\underline{\theta}}^{\overline{\theta}}\pi_{t}(\theta)\theta dF}{\int_{\underline{\theta}}^{\overline{\theta}}\pi_{t}(\theta) dF}(1+\epsilon).
    $$Recalling that $m=\frac{\int_{\underline{\theta}}^{\overline{\theta}}\pi_{t}(\theta)\theta dF}{\int_{\underline{\theta}}^{\overline{\theta}}\pi_{t}(\theta) dF}$, this is what we needed to show.
\end{proof}

\subsection{Claim of equilibrium existence in \cref{sec_wage}}\label{app:example}

If $D=[0,1]$, then any equilibrium has $d=1$ and the equilibrium wage is $$
w^*=\frac{\int_0^1\theta^1\theta d\theta}{\int_0^1\theta^1 d\theta}=\frac{2}{3}.
$$

Take any deviation to $d< 1$ and $m(h)\leq 2/3$ (by the negative selection, any wage higher would result in negative profits). 

The cutoff type choosing the deviating firm is $$
\theta^*(d,w)\text{ solving }(\theta^*)^d w=\theta^* \frac{2}{3}\;\Rightarrow\; \theta^*(d,w)=\left(\frac{3w}{2}\right)^{1/(1-d)}.
$$
Given that $w\leq 2/3$, $\theta^*(d,w)\in [0,1]$. Any deviating firm's payoffs are $$
\int_0^{\theta^*(d,w)}\theta^d(\theta-w)d\theta.
$$Calculating this integral, we obtain$$
(\theta^*)^{d+1}\left(\frac{\theta^*}{d+2}-\frac{w}{d+1}\right).
$$We can use the identity $w=\frac{2}{3}(\theta^*)^{1-d}$, to get$$
(\theta^*)^{d+2}\left(\frac{1}{d+2}-\frac{2}{3(d+1)}(\theta^*)^{-d}\right).
$$Using that $\theta^*\leq 1$ and $d< 1$, it is easy to verify that this expression is negative and therefore no deviation is profitable.

\section{Isocost tests with difficulty comparison}\label{app_cost}

For a given test $t$, let $\overline{\pi}_t=\int_\Theta \pi_t(\theta)dF$. We define the cost as follows. Let $f_{th}(\theta)=f(\theta)\frac{\pi_t(\theta)}{\overline{\pi}_t}$ and $f_{tl}(\theta)=f(\theta)\frac{1-\pi_t(\theta)}{1-\overline{\pi}_t}$. I assume that the cost associated with test $t$ is posterior separable \citep{caplin_et_al2022}:$$
C(t)=\overline{\pi}_t c(f_{th})+(1-\overline{\pi}_t)c(f_{tl}),
$$where $c:\Delta \Theta\rightarrow \reals$ is a strictly convex and continuous function.\footnote{I endow the set $\Delta\Theta$ with the weak* topology. In this topology, a sequence $(F_n)_n$ converges to $F$ if for any continuous $\phi:\Theta\rightarrow\reals$, $\int\phi(\theta)dF_n(\theta)\rightarrow\int\phi(\theta)dF$.} This class of cost function includes many commonly used cost functions such as mutual information cost or log-likelihood ratio cost \citep{SIMS2003,pomatto_et_al2023}. 

\begin{proposition}
    Let $\pi_t(\theta)\in (0,1)$ for all $\theta\in \Theta$. Then there exists a test $t'\neq t$ with $t\succeq_d t'$ and $C(t)=C(t')$.
\end{proposition}

\begin{proof}
    Let $\mu\in (0,\frac{1-\pi_t(\overline{\theta})}{1-\overline{\pi}_t})$, $\lambda\in [0,1]$ and $\pi_{t'}(\theta)=\frac{\lambda\overline{\pi}_t+(1-\lambda)\pi_t(\theta)}{(1-\mu)(1-\lambda(1-\overline{\pi}_t))+\mu\overline{\pi}_t}$. Suppose there is a symmetric equilibrium $s=(t,\alpha)$ with $\int_\Theta\theta \pi_t(\theta)dF>0$. Let $\mu\in [0,\frac{1-\pi_t(\overline{\theta})}{1-\overline{\pi}_t})$ (such a $\mu$ exists because $\pi_t(\overline{\theta})<1$). Define the test $t'$ as follows: $$
\pi_{t'}(\theta)=\frac{\lambda\overline{\pi}_t+(1-\lambda)\pi_t(\theta)}{(1-\mu)(1-\lambda(1-\overline{\pi}_t))+\mu\overline{\pi}_t}.
    $$One can check that with the assumption on $\mu$, we always have $\pi_{t'}(\theta)\leq 1$. 
    
    Monotonicity of $\pi_{t'}$ follows from the definition and the monotonicity of $\pi_t$.

    Let $\overline{\pi}_{t'}=\int_\Theta\pi_{t'}(\theta)dF$. Under our assumption on $\mu$, we get $\overline{\pi}_{t'}=\frac{\overline{\pi}_t}{(1-\mu)(1-\lambda(1-\overline{\pi}_t))+\mu\overline{\pi}_t}$. Note as well that for all $\theta\in \Theta$, we have\begin{align*}
        &\lambda f(\theta)+(1-\lambda)f_{th}(\theta)=f_{t'h}(\theta),\\
        &\mu f(\theta)+(1-\mu)f_{t'l}(\theta)=f_{tl}(\theta).
    \end{align*}
    These expressions can be easily verified by plugging in the values of $\pi_{t'}$ and $\overline{\pi}_{t'}$.

    This implies that \begin{align*}
        &\lambda+(1-\lambda)\frac{\pi_t(\theta)}{\overline{\pi}_t}=\frac{\pi_{t'}(\theta)}{\overline{\pi}_{t'}}\\
        \text{and }&\mu+(1-\mu)\frac{1-\pi_{t'}(\theta)}{1-\overline{\pi}_{t'}}=\frac{1-\pi_t(\theta)}{1-\overline{\pi}_t}.
    \end{align*}If $\pi_{t'}(\theta)\in (0,1)$, one can check from these expressions that $\frac{\pi_t(\theta)}{\pi_{t'}(\theta)}$ and $\frac{1-\pi_t(\theta)}{1-\pi_{t'}(\theta)}$ are increasing using that $\pi_{t'}$ is increasing. Therefore, $t\succeq_d t'$.

    We want to show there exists $\lambda\in (0,1)$ such that $$
    \overline{\pi}_{t'} c(f_{t'h})+(1-\overline{\pi}_{t'})c(f_{t'l})=\overline{\pi}_t c(f_{th})+(1-\overline{\pi}_t)c(f_{tl}).
    $$Note that if $\lambda=1$, then $\pi_{t'}(\theta)=1$ for all $\theta$ and $f_{t'h}=f$. Therefore, $t'$ is uninformative and $C(t')<C(t)$. If on the other hand, $\lambda=0$, we obtain $\overline{\pi}_{t'}=\frac{\overline{\pi}_t}{1-\mu+\mu \overline{\pi}_t}$ and $f_{t'h}=f_{th}$. To apply the intermediate value theorem and prove our claim, we want to show that $$
    \frac{\overline{\pi}_t}{1-\mu+\mu \overline{\pi}_t}c(f_{th})+(1-\frac{\overline{\pi}_t}{1-\mu+\mu \overline{\pi}_t})c(f_{t'l})>\overline{\pi}_tc(f_{th})+(1-\overline{\pi}_t)c(f_{tl}).
    $$We can use the fact that $c$ is convex and $f_{tl}=\mu f+(1-\mu)f_{t'l}$, to strengthen this inequality to $$
    \frac{\overline{\pi}_t}{1-\mu+\mu \overline{\pi}_t}c(f_{th})+(1-\frac{\overline{\pi}_t}{1-\mu+\mu \overline{\pi}_t})c(f_{t'l})>\overline{\pi}_tc(f_{th})+(1-\overline{\pi}_t)(\mu c(f)+(1-\mu)c(f_{t'l})).
    $$Using that $c(f)=0$ and rearranging, we obtain$$
    \overline{\pi}_tc(f_{th})+(1-\mu)(1-\overline{\pi}_t)c(f_{t'l})>0,
    $$which is satisfied. Therefore, by continuity of $c$, there is $\lambda\in (0,1)$ that delivers $C(t')=C(t)$.
\end{proof}

\end{document}